\RequirePackage{xcolor}
\documentclass[journal,twoside,web]{ieeecolor_mod}

\usepackage{generic}
\usepackage{cite}
\usepackage{amsmath,amssymb}
\usepackage{graphicx}
\usepackage{algorithm,algorithmic}

\usepackage[hidelinks]{hyperref}
\usepackage{textcomp}
\usepackage{stmaryrd}
\usepackage{tikz,tikz-cd}

\newtheorem{theorem}{Theorem}[section]
\newtheorem{lemma}[theorem]{Lemma}
\newtheorem{definition}{Definition}
\newtheorem{corollary}[theorem]{Corollary}
\newtheorem{proposition}[theorem]{Proposition}

\newtheorem{remark}{Remark}
\newtheorem{assumption}{Assumption}
\newtheorem{example}[theorem]{Example}

\allowdisplaybreaks[4] 

\newcommand{\real}{\mathbb{R}}

\newcommand{\e}{\operatorname{e}}

\newcommand{\tx}{{\tilde{x}}}
\newcommand{\tz}{{\tilde{z}}}

\newcommand{\osLip}{\mathsf{osLip}}
\newcommand{\Lip}{\mathsf{Lip}}

\newcommand{\sfG}{\mathsf{G}}
\newcommand{\sfF}{\mathsf{F}}

\newcommand{\sfGn}{\mathsf{G}^n}

\newcommand{\solxone}{\ensuremath{x}}
\newcommand{\solxtwo}{\ensuremath{\bar{x}}}
\newcommand{\solzone}{\ensuremath{z}}
\newcommand{\solztwo}{\ensuremath{\bar{z}}}

\newcommand{\solyone}{\ensuremath{y}}
\newcommand{\solytwo}{\ensuremath{\bar{y}}}

\DeclareMathOperator*{\argmin}{arg\,min}

\newcommand*\circled[1]{\tikz[baseline=(char.base)]{
\node[shape=circle,draw,inner sep=1pt] (char) {#1};}}

\definecolor{gnblue6}{RGB}{35,156,255} 

\definecolor{subsubsection-blue}{RGB}{0,80,160}
\newcommand{\subblue}[1]{{\color{subsubsection-blue} #1}}

\usepackage{epstopdf}
\graphicspath{{./IMG/}{images/}{./IMG/SAGEMATH_PLOTS/}}

\def\BibTeX{{\rm B\kern-.05em{\sc i\kern-.025em b}\kern-.08em
    T\kern-.1667em\lower.7ex\hbox{E}\kern-.125emX}}

\begin{document}

\title{Sampled-data Systems: Stability,
Contractivity and Single-iteration Suboptimal MPC}

\author{Yiting Chen, Francesco Bullo, and Emiliano Dall'Anese 
\thanks{Y. Chen is with the Department of Electrical and Computer Engineering, Boston University, Boston, MA 02215, USA. Email: \texttt{yich4684@bu.edu}.}
\thanks{F. Bullo is with the Center for Control, Dynamical Systems, and Computation, UC Santa Barbara, Santa Barbara, CA 93106 USA. Email: \texttt{bullo@ucsb.edu}.}
\thanks{E. Dall'Anese is with the Department of Electrical and Computer Engineering and the Division of Systems Engineering,  Boston University, Boston, MA 02215, USA. Email: \texttt{edallane@bu.edu}.}
\thanks{This work was supported in part by the NSF Award 2444163 and the AFOSR Awards AFOSR FA9550-22-1-0059 and FA9550-23-1-0740.} 
\thanks{The authors thank Dr.\ Emilio Benenati for  insightful discussions.}
}

\maketitle

\begin{abstract}
  This paper analyzes the stability of interconnected continuous-time (CT)
  and discrete-time (DT) systems coupled through sampling and zero-order
  hold mechanisms. The DT system updates its output at regular intervals
  $T>0$ by applying an $n$-fold composition of a given map.  This setup is
  motivated by online and sampled-data implementations of
  optimization-based controllers -- particularly model predictive control
  (MPC) -- where the DT system models $n$ iterations of an algorithm
  approximating the solution of an optimization problem.

  We introduce the concept of a \emph{reduced model}, defined as the
  limiting behavior of the sampled-data system as $T \to 0^+$ and $n \to
  +\infty$.  Our main theoretical contribution establishes that when the
  reduced model is contractive, there exists a threshold duration $T(n)$
  for each iteration count $n$ such that the CT-DT interconnection achieves
  exponential stability for all sampling periods $T < T(n)$.  Finally,
  under the stronger condition that both the CT and DT systems are
  contractive, we show exponential stability of their interconnection using
  a small-gain argument.  Our theoretical results provide new insights into
  suboptimal MPC stability, showing that convergence guarantees hold even
  when using a single iteration of the optimization algorithm—a practically
  significant finding for real-time control applications.
\end{abstract}

\section{Introduction}
\label{sec:introduction}

This paper studies the interconnection of a continuous-time (CT) dynamical system with a discrete-time (DT) system, coupled via sampling and zero-order hold mechanisms. As illustrated in the block diagram  in Fig.~\ref{fig:block-diagram} \emph{(Left)}, the DT system acquires samples of the state $x$ regularly at intervals of length $T$, and updates its output  $z$ based on a map $\sfGn(x,z)$; in particular,  $\sfGn(x,z)$ is defined as the $n$-fold composition of a given map $\sfG(x,z)$. In this setup, the goal of this paper is threefold: \emph{(i)} First, inspired by the foundational work on two-time-scale CT systems~\cite{kokotovic1968singular,kokotovic1976singular}, we aim to formalize the notion of a \emph{reduced model} (RM) associated with~\eqref{eq:sample_data} by  viewing $T$ and $n$ as parameters that induce a ``time-scale separation''; a preview of our RM model is in Fig.~\ref{fig:block-diagram} \emph{(Right)}, where  $z^*(x) := \lim_{n\to +\infty }\sfGn(x,z)$ for any fixed $x$. \emph{(ii)}~We examine how the properties of the RM, $n$, and $T$ influence the stability of the interconnected system. These conditions are then compared with those derived from small-gain arguments. Our analysis builds on tools from contraction theory~\cite{lohmiller1998contraction,FB:26-CTDS}, which provide a natural framework for studying stability of the interconnected system in Fig.~\ref{fig:block-diagram} \emph{(Left)} using Lipschitz constants for the CT and DT systems. For many systems of interest, these  constants can be computed analytically or estimated numerically. Before moving onto a brief literature review, we first provide motivations for the overall setup.

\begin{figure}[t]
\centering 
\includegraphics[width=1.0\columnwidth]{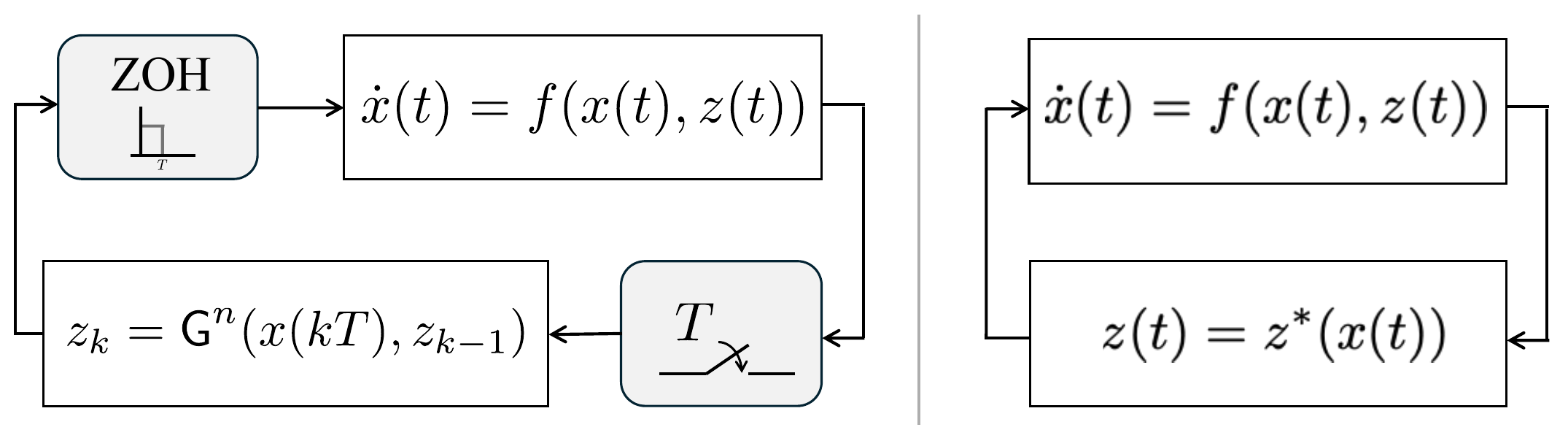}
\caption{\emph{(Left)} Interconnected system considered in this paper, with a CT sub-system, a DT sub-system, and sampling and zero-order hold blocks. Here, $\sfGn(x,z)$ is defined as $\sfG^1(x,z)=\sfG(x,z), \ldots, \sfGn(x,z) = \sfG(x,
\sfG^{n-1}(x,z))$. \emph{(Right)} Reduced model. Here, for any $x$, $z^*(x)$ is defined by $z^*(x) = \sfG(x,z^*(x))$. }
\label{fig:block-diagram}
\vspace{-.5cm}
\end{figure}

\emph{Motivations}. Our analysis is motivated by optimization-
based control frameworks -- which include, for example,  model predictive control (MPC)~\cite{rawlings2017model,gros2020linear,morari1988model}, Control Barrier Function (CBF) based control~\cite{ames2019control,xiao2023safe} and extensions including both CBF and Control Lyapunov Functions (CLFs)~\cite{garg2019control}, pointwise minimum-norm control~\cite{freeman2008robust}, and online feedback optimization~\cite{colombino2019online, hauswirth2020timescale}. These frameworks share a key characteristic: the control input is implicitly defined as the solution to an optimization problem that encodes performance objectives and system constraints. In our setup, the closed-loop system resulting from the interconnection of a plant with such an optimization-based controller is abstracted by the RM shown in Fig.~\ref{fig:block-diagram} \emph{(Right)}, where the map 
 $z^*(x)$ represents the control law obtained as the solution to the underlying optimization problem.  Accordingly, the interconnected system depicted in Fig.~\ref{fig:block-diagram} \emph{(Left)} captures practical aspects of real-world implementations, where the system state 
 $x$ is sampled at discrete time intervals, and the computation of the control input $z^*(x)$ is limited by computational constraints. Specifically, instead of solving the optimization problem to full convergence, only a finite number $n$ of iterations of an algorithm are executed. 
 
 We are particularly interested in applying our theoretical results to MPC problems, where the optimal solution map is often not available in closed form. This motivates the need for real-time MPC implementations~\cite{gros2020linear,diehl2005real,giselsson2010distributed}, as well as approaches such as instant MPC and sub-optimal MPC~\cite{zeilinger2011real,figura2024instant,yoshida2019instant,liao2020time,liao2021analysis,van2024suboptimal,moriyasu2024sampled,karapetyan2023finite}. See also~\cite{karapetyan2025closed} for a broader perspective on online control and the recent work on receding horizon games in~\cite{hall2024stability}.

\emph{Connections to the literature}. The study of two-time-scale CT systems dates back to the foundational works in singular perturbation theory~\cite{kokotovic1968singular,kokotovic1976singular}. A unified framework for analyzing stability and robustness in such systems was developed in~\cite{teel2003unified}, while contraction-based approaches were explored in~\cite{del2012contraction,cothren2023singular}. 
Contractivity of reduced-type dynamics has been studied in the 
context of nonlinear differential–algebraic equation systems in, e.g.,~\cite{nguyen2020contraction}.
  
Stability of coupled or networked CT dynamical systems was investigated in, e.g.,~\cite{russo2012contraction,picallo2022sensitivity}. 
Our analysis is inspired by these works, which prompted us to treat $T$ and $n$ as parameters that induce a time-scale separation. When considering the limiting case 
$n = + \infty$, $T > 0$, our framework encompasses classical sampled-data control systems~\cite{heemels2010stability}. The study of such systems has a rich literature, with numerous foundational contributions; representative works include~\cite{nesic2004framework,naghshtabrizi2006robust,sivashankar1993robust,fujioka2009discrete,gabriel2021sampled,fagundes2023stability}. Notably, in the context of linear time-invariant (LTI) systems,~\cite{naghshtabrizi2006robust} establishes stability conditions based on the sampling period $T$, formulated through Linear Matrix Inequalities (LMIs). In the context of instant MPC and time-distributed MPC, existing stability analysis focuses on CT-CT interconnections~\cite{figura2024instant,yoshida2019instant} or DT-DT settings~\cite{liao2020time,liao2021analysis,van2024suboptimal}.   

\emph{Contributions}. Our contribution is threefold. 

\noindent \emph{(a)} First, we formalize the notion of RM, which is defined as the limiting system for $T\!\rightarrow \!0^+$ and $n\!\rightarrow \!+\infty$, modeling the case where the map $z^*(x)$ is applied at all times. Then, we show our main result: if the DT system is contractive and the RM is strongly infinitesimally contractive, then for any $n \in \mathbb{Z}_{>0}$ there exists a $T(n)$ such that the interconnection between the CT and DT systems is globally exponentially stable (GES) for any $T\!<\!T(n)$. We offer a detailed expression for $T(n)$, which depends solely on $n$ and on  the Lipschitz constants of the RM and of the CT and DT systems. We also show that $T(n)$ is strictly increasing with 
respect to $n$ and it is bounded.   

\noindent \emph{(b)} Under the assumption that the CT system and DT system are contractive, we show that the interconnected system in Fig.~\ref{fig:block-diagram} \emph{(Left)} is globally exponentially stable and contractive in discrete time (i.e., when trajectories are sampled at intervals $T$) for any $n$ and any $T$ if the interconnection gain is small enough. 
This yields a quantitative small-gain condition for the CT–DT interconnection and extends classical small-gain contraction theorems (which apply only to CT–CT interconnections~\cite[Sec.~3.6]{FB:26-CTDS}\cite{russo2012contraction}) to the CT–DT architecture.
We show that these conditions imply contractivity of the RM and, thus, they are stricter than the one derived in \emph{(a)}. 
We briefly summarize the relations among the contributions \emph{(a)} and \emph{(b)} as an implication diagram in Fig.~\ref{fig:brief_diagram}.

\noindent \emph{(c)} 
As an application of our theoretical result in \emph{(a)}, we consider the MPC problem associated with an LTI dynamical system.
Interestingly, and in contrast
with existing results~\cite{liao2020time,liao2021analysis}, our findings
establish stability of the online MPC even when $n = 1$, provided that $T$
is sufficiently small. In other words, we show that only one evaluation
(i.e., only one iteration) of the MPC solver is sufficient to ensure
stability, provided $T$ is sufficiently small.  To the best of our
knowledge, 
this is the first time that such an MPC result is rigorously established.

\begin{figure}
    \centering
\vspace{.1cm}
\noindent \framebox[.98\width]{
\begin{tikzpicture}
[>=stealth, ->, line width=1.5pt]
\node[align=center] (vc) at (-3.2,3) {\small{\textsf{Overview of the implication diagram}} };
\node[align=center] (v2_1) at (-6,0) {\small{RM  contractive } \\
\small{DT subsystem contractive }};
\node[align=center] (v2_2) at (-1.4,0) {\small{$\forall~n\in\mathbb{Z}_{>0},~\exists~T(n)\!>\!0$: }\\ \small{$\forall~T\!<\!T(n)$, CT-DT system is GES }};  
    \node[align=center] (v1_1) at (-6,2.0) {\small{CT subsystem contractive} \\ \small{DT subsystem contractive }\\ \small{$+$ small-gain condition}};
\node[align=center] (v1_2) at (-1.4,2.0) {$\forall~n\in\mathbb{Z}_{>0}$ and $\forall~T>0$: \\ \small{CT-DT system is GES and } \\ \small{contractive (in discrete-time)} };
\draw[->]  (v1_1) edge  node[midway, above, sloped] {} (v1_2) ;
\draw[->] (v1_2) edge node[midway, right] {} (v2_2);
\draw[->]  (v2_1) edge node[midway, above, sloped] {} (v2_2);
\draw[->]  (v1_1) edge node[midway, above, right] {} (v2_1);
\end{tikzpicture}
}
 \caption{Implication structure underlying the main results for the CT-DT system. Rigorous formulations for the general nonlinear case and the LTI specialization are provided in  Sections~\ref{sec: main results} and~\ref{sec: LTI-case}, respectively.
 }
\vspace{-.5cm}
    \label{fig:brief_diagram}
\end{figure}

The rest of the paper is organized as follows. Section~\ref{sec:preliminaries} presents some mathematical preliminaries and definitions. Section~\ref{sec:ctdt-interconnection} outlines our main setup, presents the main results, and illustrates the implication diagram. Section~\ref{sec:mpc} applies our findings to the MPC, and presents some illustrative numerical results. Proofs of the results are presented in Section~\ref{sec:proofs}, while Section~\ref{sec:conclusions} concludes the paper.

\section{Mathematical Preliminaries}
\label{sec:preliminaries}

\subsection{Notation and definitions}
We denote by $\real$, $\real_{>0}$, and $\real_{\geq 0}$ the set of real numbers, positive real numbers and nonnegative real numbers, respectively. Similarly, we define $\mathbb{Z}$, $\mathbb{Z}_{>0}$, and $\mathbb{Z}_{\geq 0}$. 
For $n\in\mathbb{Z}_{>0}$,
$\mathbf{0}_n$ is the $n$-dimensional zero vector. 
We also write $[n] = \{ 1, \ldots, n \}$. Vector inequalities of the form $x \leq  (\geq,>,<) ~y$ are entry-wise. Given a matrix  $M \in \mathbb{R}^{N \times N}$, its spectrum $\operatorname{spec}(M)$ is the set of its eigenvalues and its spectral radius $\rho(M)$ and spectral abscissa $\alpha(M)$ are defined by $\rho(M): = \max\{|\lambda|\!:  \lambda \in  \operatorname{spec}(M)\}$ and
$\alpha(M):= \max\{\operatorname{Re}(\lambda)\ \! :  \lambda \in  \operatorname{spec}(M)\}$,  respectively, where $|\lambda|$ and $\operatorname{Re}(\lambda)$ denote the absolute value and the real part of $\lambda$. A matrix $M \in \mathbb{R}^{N \times N}$ is Schur stable if $\rho(M)<1$ and Hurwitz stable if $\alpha(M)<0$.
We define the weighted $\ell_p$ norm $\|\cdot \|_{p,[\eta]}$ on $\mathbb{R}^n$ as  $\|x \|_{p,[\eta]}:=\left(\sum_{i=1}^n \eta_i |x_i|^p\right)^{\frac{1}{p}}$, where $\eta=[\eta_1,..,\eta_n]^\top$.
Given  $\mathcal{X}\subseteq \mathbb{R}^n$ and  $\mathcal{Z}\subseteq \mathbb{R}^m$, we let $\| \cdot\|_{\mathcal{X}}$ and $\| \cdot\|_{\mathcal{Z}}$ be the vector norms on $\mathcal{X}$ and $\mathcal{Z}$, respectively. We say a vector norm $\|\cdot \|$ on $\mathbb{R}^n$ is monotonic if for any $v, w\in\mathbb{R}^n$, $\left|v_i\right| \leq\left|w_i\right|, i \in\{1, \ldots, n\}  \Longrightarrow \|v\| \leq\|w\|$. Next we define the matrix norm $\| \cdot\|_{\mathcal{E}_1 \to \mathcal{E}_2}$, $\mathcal{E}_1,\mathcal{E}_2\in\{\mathcal{X},\mathcal{Z}\}$, as 
    $\| M\|_{\mathcal{E}_1 \to \mathcal{E}_2}:=\max_{ \|v \|_{\mathcal{E}_1}=1 } \|Mv \|_{\mathcal{E}_2}$. 
If $\mathcal{E}_1=\mathcal{E}_2$, we simply write the matrix norm $\| \cdot\|_{\mathcal{X} \to \mathcal{X}}$ (resp. $\| \cdot\|_{\mathcal{Z} \to \mathcal{Z}}$ ) as $\| \cdot\|_\mathcal{X} $ (resp. $\| \cdot\|_\mathcal{Z} $), since they are exactly the matrix norms induced by the corresponding vector norms. We define the induced logarithmic norm (log-norm) $\mu_{\| \cdot\|_\mathcal{X}}(\cdot)$ of matrix $M$ by:
$$
\mu_{\| \cdot\|_\mathcal{X}}(M)=\lim _{h \rightarrow 0^{+}} \frac{\left\|I_n+h M\right\|_\mathcal{X}-1}{h}.
$$

We refer to \cite[Sec.~2.6]{FB:26-CTDS} for the definition and properties of weak pairings. We let $\llbracket \cdot, \cdot
\rrbracket_\mathcal{X}$ denote a weak pairing on $\mathcal{X}$ compatible with the norm $\|\cdot\|_\mathcal{X}$ (such a weak pairing always exists).  Since we only need log-norm and
weak pairing on $\mathcal{X}$, we will omit the sub-index for log-norm and weak pairing throughout this paper.

Given two normed spaces $\left(\mathcal{E}_1,\|\cdot\|_{\mathcal{E}_1}\right),\left(\mathcal{E}_2,\|\cdot\|_{\mathcal{E}_2}\right)$ and  a norm $\|\cdot\|$ on $\mathbb{R}^2$, we define a composite norm $\|\cdot\|_{\textsf{cmp}}$ on $\mathcal{E}_1\times \mathcal{E}_2$ by $\| 
[w^\top,v^\top]^\top \|_{\textsf{cmp}}:=\| [\|w\|_{\mathcal{E}_1}, \|v\|_{\mathcal{E}_2}   ]^\top  \|$, $\forall~w\in\mathcal{E}_1,~\forall~v\in \mathcal{E}_2$~\cite[Sec.~2]{FB:26-CTDS},~\cite{russo2012contraction}.
A map $\sfG: \mathcal{E}_1 \rightarrow \mathcal{E}_2$ is Lipschitz with constant $L \geq 0$ if $\left\|\sfG\left(w_1\right)-\sfG\left(w_2\right)\right\|_{\mathcal{E}_2} \leq L\left\|w_1-w_2\right\|_{\mathcal{E}_1}$, for all $w_1, w_2 \in \mathcal{E}_1$. Furthermore, if $\mathcal{E}_1=\mathcal{E}_2$, the map is one-sided Lipschitz with constant $\Tilde{L}$ if $\llbracket \sfG\left(w_1\right)-\sfG\left(w_2\right) ; w_1-w_2 \rrbracket_{\mathcal{E}_1} \leq \Tilde{L}\|w_1-w_2\|_{\mathcal{E}_1}^2$ for all $w_1, w_2 \in \mathcal{E}_1$. The minimal Lipschitz constant and the minimal one-sided Lipschitz constant of the map $w \mapsto \sfG(w)$ are defined by
$$
\begin{aligned}
\Lip_w(\sfG) & :=\sup _{w_1 \neq w_2} \frac{\|\sfG(w_1)-\sfG(w_2)\|_{\mathcal{E}_2}}{\|w_1-w_2\|_{\mathcal{E}_1}} \in \mathbb{R}_{\geq 0} \\
\osLip_w(\sfG) & :=\sup _{w_1 \neq w_2} \frac{\llbracket \sfG(w_1)-\sfG(w_2) ; w_1-w_2 \rrbracket_{\mathcal{E}_1}}{\|w_1-w_2\|^2_{\mathcal{E}_1}} \in \mathbb{R} .
\end{aligned}
$$
Similarly, we say that the map $v\mapsto\sfF(w,v)$ is uniformly Lipschitz over $w$ if $\Lip_v(\sfF)<+\infty$, where 
\begin{align*}
\Lip_v(\sfF) & := \sup_w  \sup _{v_1 \neq v_2} \frac{\|\sfF(w,v_1)-\sfF(w,v_2)\|}{\|v_1-v_2\|}, 
\end{align*}
where the norms are defined appropriately~\cite[Ch.~2]{FB:26-CTDS}.
The map $v\mapsto\sfF(w,v)$ is uniformly one-sided Lipschitz over $w$ if $\osLip_v(\sfF)<+\infty$, where
\begin{align*}
    \osLip_v(\sfF) & :=\sup_w \sup _{v_1 \neq v_2} \frac{\llbracket \sfF(w,v_1)-\sfF(w,v_2) ; v_1-v_2 \rrbracket}{\|v_1-v_2\|^2},
\end{align*}
and the weak pairing and norm are defined consistently.

\subsection{Contractivity and stability of dynamical systems}

Consider a dynamical system $\dot x=\sfF(x)$, where $\sfF:  \mathcal{C} \mapsto \mathcal{C}$ is  continuous. Assume that $\mathcal{C} \subseteq \mathbb{R}^n$ is convex, and that it is forward invariant for the dynamics $\dot x=\sfF(x)$. Let $t \mapsto \phi (t;x_0)$ be the flow map of the system from an initial point $x(0) := x_0 \in \mathcal{C}$. We give the following definitions of stability and contractivity (and refer the reader to~\cite[Ch.~4]{khalil2002nonlinear} and~\cite[Ch.~3]{FB:26-CTDS} for background and additional details).

\vspace{.1cm}

 \begin{definition}[Exponential stability]
 \label{def:ges}
      Let $x^* \in \mathcal{C}$ be such that $\sfF(x^*)=0$. The system $\dot x=\sfF(x)$ renders the  equilibrium $x^* \in \mathcal{C}$ globally exponentially stable if there exists a norm $\|\cdot\|$ defined on $\mathcal{C}$ and constants $a  >0$ and $r> 0$ such that  $\|\phi (t;x_0)-x^* \|\leq r \e^{-at}\|x_0-x^* \|, ~\forall~t \geq 0$, and $\forall~x_0 \in \mathcal{C}$. \hfill $\Box$
 \end{definition}

\vspace{.1cm}
 
 \begin{definition}[Strongly infinitesimally contracting system]
  \label{def:contraction}
    Given a norm $\|\cdot\|$ on $\mathbb{R}^n$ with a compatible weak pairing $\llbracket \cdot ; \cdot \rrbracket $, the vector field $\sfF$ is strongly infinitesimally contracting on $\mathcal{C}$ with 
    contraction rate $c>0$ if  $\osLip_x(\sfF)\leq -c$. \hfill $\Box$
\end{definition}

\vspace{.1cm}

By~\cite[Theorem~3.7]{FB:26-CTDS}, the vector field $\sfF$ is $c$-strongly
infinitesimally contracting if and only if any two trajectories $\phi
(t;x_0)$ and $\phi (t;y_0)$ of the system $\dot x=\sfF(x)$ satisfy $\|\phi
(t;x_0)-\phi (t;y_0) \|\leq \e^{-ct}\|x_0-y_0 \|$, for all $t \geq 0$. In
addition, if $\sfF$ is $c$-strongly infinitesimally contracting, then there
exists a unique equilibrium $x^*$ that is globally exponentially stable at
rate $c$.

For a given $T \in \mathbb{R}_{>0}$, consider the sampled trajectories  $\phi (kT;x_0)$ and $\phi (kT;y_0)$, $k \in \mathbb{N}$,  of the system $\dot x=\sfF(x)$ for initial conditions $x_0, y_0 \in \mathcal{C}$. We state the following definition. 

\vspace{.1cm}

 \begin{definition}[Discrete-time contracting flow map]
  \label{def:DTcontraction}
    Given a norm $\|\cdot\|$ on $\mathbb{R}^n$ and $T>0$, the continuous-time system $\dot x=\sfF(x)$ is discrete-time contracting on $\mathcal{C}$ if 
   $\exists~b \in [0,1)$ such that  $\|\phi (kT;x_0) - \phi (kT;y_0)\| \leq b^k \|x_0 - y_0\|$  for all $k \in \mathbb{N}$ and $x_0, y_0 \in \mathcal{C}$. \hfill $\Box$
\end{definition}

\vspace{.1cm}

Definition~\ref{def:DTcontraction} considers samples of a trajectory $\phi (t;x_0;y_0)$ at times $t = kT, k \in \mathbb{N}$. This is a key distinction relative to Definition~\ref{def:contraction}, since the transient bound applies to only the time instants $t = kT, k \in \mathbb{N}$, and not on the distance between trajectories $\phi (\tau + kT;x_0) - \phi (\tau + kT;y_0)$ during the open interval  $\tau + kT, \tau \in (0,T)$. Finally, we have the following definition for discrete-time dynamical systems.

\vspace{.1cm}

 \begin{definition}[Contracting discrete-time system]
  \label{def:contraction-DTsystem}
    Consider a discrete-time dynamical system $x_{k+1} =\sfG(x_{k})$, $k \in \mathbb{N}$, where $\sfG:  \mathbb{R}^n \mapsto \mathbb{R}^n$ is  Lipschitz continuous. Given a norm $\|\cdot\|$ on $\mathbb{R}^n$, the system $x_{k+1} =\sfG(x_{k})$ is contracting (in discrete time) if $\Lip_x(\sfG) <1$.  \hfill $\Box$
\end{definition}

\vspace{.1cm}

Hereafter, we use the shorthand terms CT and DT to refer to \emph{continuous-time} and \emph{discrete-time} systems, respectively.

\section{Interconnection of a Continuous-time System and a Discrete-time System}
\label{sec:ctdt-interconnection}

\subsection{Model setup, assumptions, and research goals}

We consider the interconnection of a CT system and of a DT system of the form (see also Fig.~\ref{fig:block-diagram}):
\begin{subequations}
\label{eq:sample_data}
\begin{align}
    \dot x(t) & = f(x(t),z(t)) \label{eq:sample_data-plant}\\
    z_{k} & = \sfGn(x(kT),z_{k-1}),  \label{eq:sample_data-discrete} \\
   z(t) & = z_{k} , \,\, t \in [kT, (k+1)T) \label{eq:sample_data-z}
\end{align}
\end{subequations}
where $t \in \mathbb{R}_{\geq 0}$ denotes time, $k \in \mathbb{Z}_{\geq 0}$ is the index for the updates of the DT sub-system,  $T > 0$ is a given time interval, $x \in \mathcal{X}\subseteq\mathbb{R}^{n_x}$ and $z \in\mathcal{Z}
\subseteq\real^{n_z}$ are the states of the CT sub-system and of the  DT sub-system, respectively, $\mathcal{X}$ and $\mathcal{Z} $ are convex and forward invariant  for~\eqref{eq:sample_data}, and   $f: \mathcal{X} \times \mathcal{Z} \to \mathcal{X}$ is continuous in its arguments; finally $\sfGn(x,z)$ is defined as $\sfG^1(x,z)=\sfG(x,z), \ldots, \sfGn(x,z) = \sfG(x,
\sfG^{n-1}(x,z))$, $n\in\mathbb{Z}_{>0}$, with $\sfG(x,z)$ continuous in its arguments. In other words, $\sfGn(x,z)$ is given by the $n$-fold composition of the map $z \mapsto \sfG(x,z)$, for fixed $x$.  Hereafter, we denote the (full) state of the interconnected system \eqref{eq:sample_data} as $y(t) := [x(t)^\top, z(t)^\top]^\top$.

We first clarify how  Definitions~\ref{def:ges} and \ref{def:DTcontraction}  apply to the interconnected system~\eqref{eq:sample_data} (as on the right-hand side of the diagram in Fig.~\ref{fig:brief_diagram}).
Following Definition~\ref{def:ges}, the interconnected system~\eqref{eq:sample_data} renders the  origin GES if, for any norm $\|\cdot\|$  on $\mathcal{X}\times\mathcal{Z}$, there exist constants $a > 0$ and $r> 0$ such that 
\begin{equation}\label{eq: def_GES}
    \| y(t)\| \leq r\e^{-at}   \| y(0)\|
\end{equation}
for any $y(0) \in \mathcal{X}\times\mathcal{Z}$.
For a given $T>0$, we say that the  interconnected system~\eqref{eq:sample_data} is $T$-discrete-time \emph{contractive} (in short, $T$-DTC) if there exists a norm $\|\cdot \|$ on $\mathcal{X}\times\mathcal{Z}$ and a constant $b \in (0,1)$ such that,  for any two solutions $\solyone(t)$ and $\solytwo(t)$ and any $k\in\mathbb{Z}_{>0}$, it holds that 
\begin{equation}\label{eq: def_contraction_dt}
    \|  \solyone(kT) - \solytwo(kT)\| \leq b^k \| \solyone(0) - \solytwo(0)\|
\end{equation}
for any $\solyone(0), \solytwo(0) \in \mathcal{X}\times\mathcal{Z}$ (cf.~Definition~\ref{def:DTcontraction}).

To set the foundation for the subsequent discussion, we present the main working assumptions adopted in this paper.

\vspace{.1cm}

\begin{assumption}[Lipschitz CT dynamics]\label{as: existence-of-weak-pairing}
   The map $x \mapsto f(x,z)$ is uniformly Lipschitz
   over $\mathcal{Z}$.  \hfill $\Box$
\end{assumption}

\vspace{.1cm}

\begin{assumption}[Lipschitz interconnection]
\label{as:Lipinterconnection-2}
The map $z \mapsto f(x,z)$  is uniformly Lipschitz (over $x \in \mathcal{X}$)  with respect to $\|\cdot\|_\mathcal{X}$ and $\|\cdot\|_\mathcal{Z}$ with constant $\Lip_z(f) > 0$. 
Similarly, the map $x \mapsto \sfG(x,z)$  is uniformly Lipschitz (over $z \in \mathcal{Z}$) with constant $\Lip_x(\sfG) > 0$.
\hfill $\Box$
\end{assumption}

\vspace{.1cm}

\begin{assumption}[Contractive DT dynamics]
\label{as:contrDiscr-2}
The map $z \mapsto \sfG(x,z)$ is uniformly Lipschitz  with constant $0<\Lip_z(\sfG) < 1$ with respect to $\|\cdot\|_\mathcal{Z}$. \hfill $\Box$
\end{assumption}

\vspace{.1cm}

Under Assumption~\ref{as: existence-of-weak-pairing}, the map $x \mapsto f(x,z)$ is uniformly one-sided Lipschitz with respect to weak pairing  $\llbracket \cdot ; \cdot \rrbracket_\mathcal{X} $ and norm $\|\cdot\|_{\mathcal{X}}$; equivalently, $ \osLip_x(f)$ is finite. Due to the Picard-Banach-Caccioppoli theorem~\cite{banach1922operations,caccioppoli1930teorema},
Assumption~\ref{as:contrDiscr-2} implies  that for any $x\in\mathcal{X}$, there exists a unique fixed point $z^*(x)$ for the map $\sfG(x,z)$;  i.e., $z^*(x)=\sfG(x,z^*(x))$. Moreover, for any fixed $x$, one has that $\lim_{n\to +\infty }\sfGn(x,z)=z^*(x)$. With this in mind, the DT system~\eqref{eq:sample_data-discrete} models the execution of $n$ Picard iterations, utilized to approximately calculate $z^*(x)$. Without loss of generality, and to streamline notation, we assume that $f(\mathbf{0}_{n_x},\mathbf{0}_{n_z})=\mathbf{0}_{n_x}$, $\sfG(\mathbf{0}_{n_x},\mathbf{0}_{n_z})=\mathbf{0}_{n_z}$, and $z^*(\mathbf{0}_{n_x}) = \mathbf{0}_{n_z}$.\footnote{If the equilibrium of the interconnected system~\eqref{eq:sample_data} is not the origin, all the arguments in this paper can be adopted upon performing a variable shift (see the standard discussion in, e.g.,~\cite[Ch.~4]{khalil2002nonlinear} and~\cite[Ch.~3]{FB:26-CTDS}).}

To ground our modeling approach, we describe its application to optimization-based control in the following remarks. 

\vspace{.1cm} 

\begin{remark}[Link to optimization-based control]
\label{example:optimization-based-control} In several optimization-based control frameworks such as MPC~\cite{rawlings2017model,gros2020linear}, CBF-based safety filters~\cite{ames2019control,xiao2023safe} and CLF-CBF extensions~\cite{garg2019control},  pointwise minimum norm control~\cite{freeman2008robust}, and online feedback optimization~\cite{colombino2019online, hauswirth2020timescale}, the feedback control law is implicitly defined as the optimal solution map of an optimization problem associated with the sub-system~\eqref{eq:sample_data-plant} (which represents a plant to be controlled). The map $z^*(x)$ in our framework represents the optimal solution map of an optimization-based controller.

  The sub-system~\eqref{eq:sample_data-discrete} is designed to capture the practical settings where: \emph{i)} the state of the system may be measured at given intervals $\{kT, k \in \mathbb{Z}_{\geq 0}\}$; and \emph{ii)} the optimization problem used to compute $z^*(x)$ is solved by an iterative algorithm and may not be solved \emph{to convergence} (i.e., only a finite number $n$ of algorithmic iterations are performed); this is because of computational complexity or computational time  considerations~\cite{liao2021analysis,van2024suboptimal}, or because solving the problem requires a pervasive sensing of disturbances~\cite{colombino2019online}. Moreover, the subsystem~\eqref{eq:sample_data-discrete} reduces to the case where $z^*(x)$ is available in closed-form by setting $\mathsf{G}(x,z)\equiv z^*(x)$, for any $z\in\mathcal{Z}$.  We  provide a more in-depth discussion of MPC implementations in Section~\ref{sec:mpc}.
  
\hfill $\Box$
\end{remark}

\vspace{.1cm}

Given this setup, the goal of this paper is threefold:

\noindent \emph{(Research goal \#1)} First, inspired by classical formulations of two-time-scale CT-CT systems~\cite{kokotovic1968singular}, we aim to formalize the notion of a \emph{reduced model} (RM) associated with~\eqref{eq:sample_data}. The parameters $T$ and $n$ are interpreted as inducing a ``time-scale separation'' between the subsystems — with $T$ capturing the frequency of updates of $z(t)$, and 
$n$ reflecting how the computation of $z^*(x)$
 is distributed over time.

\noindent \emph{(Research goal \#2)} A second goal is to unveil conditions on the RM and the CT subsystem~\eqref{eq:sample_data-plant} to ensure contractivity or stability of the interconnected system~\eqref{eq:sample_data} (in the sense of Definitions~\ref{def:ges}--\ref{def:DTcontraction}), while also elucidating how these conditions depend on the parameters $T$ and $n$.

\noindent \emph{(Research goal \#3)} Finally, we aim to extend classical  network contraction theorems for interconnected CT systems~\cite[Sec.~3.6]{FB:26-CTDS} to the system~\eqref{eq:sample_data}, and to compare the resulting conditions with those obtained in (ii).

\vspace{-0.1cm}
\subsection{The reduced model}

We define the \emph{reduced model}  associated with the interconnected system~\eqref{eq:sample_data} as:
\vspace{-0.1cm}
\begin{equation}\label{eq: reduced dynamics}    
   \textsf{RM:~~~~} \dot x(t)=f(x(t),z^*(x(t))) ,  \,\,\,\,\, t \in \mathbb{R}_{\geq 0}, 
\end{equation}
$x(0) \in \mathcal{X}$. Comparing~\eqref{eq: reduced dynamics} and~\eqref{eq:sample_data}, and in light of Assumption~\ref{as:contrDiscr-2},~\eqref{eq: reduced dynamics} can be interpreted as a limiting system  where: 
\begin{itemize}
    \item $n \rightarrow + \infty$, so that $z^*(x)$ is computed for any $x$; and
    \item $T \rightarrow 0^+$, so that $z^*(x)$ is applied at all times. 
\end{itemize}

To further motivate the \emph{(Research goals \#2 and \#3)} and to highlight challenges and the inherent limitations imposed by discrete sampling, we start by  examining the role of $T$, which determines the update frequency of $z(t)$. While the map $x \mapsto z^*(x)$ may be designed so that the reduced model
$f(x,z^*(x))$ enjoys stability or contractivity properties under
continuous-time feedback, such properties are not necessarily preserved under sampled-data
implementations, thereby requiring restriction on $T$.
To see this, we consider the system
\vspace{-0.4cm}
\begin{subequations}
\label{eq:rm-sampled}
\begin{align}
    \dot{x}(t) & = f(x(t),z(t)) \\
    z(t) & = z^*(x(kT)) ,  \,\, t \in [kT, (k+1)T)
\end{align}
\end{subequations}
for some $x(0) \in \mathcal{X}$ and $z(0) \in \mathcal{Z}$. 
This is a version of~\eqref{eq:sample_data} where $n \rightarrow + \infty$ but $T > 0$.  The following simple example shows that~\eqref{eq:rm-sampled} may be unstable if $T$ is too large.

\vspace{.1cm} 

\begin{example}[Large $T$ may not preserve stability]
\label{example:Low-dimensional-LTI}
Consider a special case of \eqref{eq:sample_data} where $\dot x=ax+bz$, $z_k=\sfGn(x_k,z_{k-1})$, $\sfG(x,z)=cx+dz$, where $a$, $b$, $c$, $d$, $x$, $z$ $\in\mathbb{R}$. Suppose that $a> 0$ and, based on Assumption~\ref{as:contrDiscr-2}, that  $|d|<1$. Then, it follows that $z^*(x)=c/(1-d)x$.  Assume that $a+bc/(1-d)<0$; this ensures that the RM  $\dot x=(a+bc/(1-d))x$ renders the origin GES. As in system~\eqref{eq:sample_data}, suppose now that $z^*(x)$ is updated every $T > 0$ due to underlying computational and sampling constraints; the dynamics are then $\dot x(t)=a x(t)+ bz(t)$, $z(t) = \frac{c}{1-d}x((k-1)T)$ for $t\in[(k-1)T,kT)$. Note that:
\begin{align*}
    x(kT)&=\e^{aT}x((k-1)T)+\int_{(k-1)T}^{kT} \e^{a(kT-\tau) } bz ~d\tau\\
&=\left(\e^{aT}+\frac{bc}{1-d}\frac{\e^{aT}-1}{a}\right) x((k-1)T)\\
&=\left(\e^{aT}+\frac{bc}{1-d}\frac{\e^{aT}-1}{a}\right)^k x(0).
\end{align*}
Since $a+\frac{bc}{1-d}<0$ and $a>0$, it follows that $\frac{bc}{a(1-d)}<-1$.
Then, for any $T\geq \frac{1}{a}\log \left( \frac{-2a(1-d)+bc}{a(1-d)+bc} \right)$, we have $\left(\e^{aT}+\frac{bc}{1-d}\frac{\e^{aT}-1}{a}\right)\leq-2$, implying that the system $\dot x(t)=ax(t)+bz(t)$, $z(t) = \frac{c}{1-d}x((k-1)T)$ for $t\in[(k-1)T,kT)$ is unstable. \hfill $\Box$
\end{example}

\vspace{.1cm} 

This example shows that, even if the RM
$f(x,z^*(x))$  is designed to exhibit desirable stability or contractivity properties in the sense of Definitions~\ref{def:ges}--\ref{def:contraction}, these properties may not be preserved if $T$ is too large. We also note that there is no universal bound on $T$ to guarantee stability: a counterexample can be easily constructed, similar to Example \ref{example:Low-dimensional-LTI}: for any $T>0$, we let $a=\frac{2}{T}$, $|d|<1$ and  $a+\frac{bc}{1-d}=-a$; it then follows that  $\left(\e^{aT}+\frac{bc}{1-d}\frac{\e^{aT}-1}{a}\right)=2-\e^2<-2$ and such system is unstable.  Therefore, a bound on the sampling time required to preserve stability must depend on the properties of $f$ and $\sfG$. In the ensuing Section~\ref{sec: main results}, we will examine the role of $T$ alongside $n$.

\vspace{.1cm} 

\begin{remark}[Connections with classical CT-CT systems]
Consider a classical two-time-scale CT system of the form~\cite{kokotovic1968singular}: 
\begin{equation}
 \label{eq:two-time-system-ct}
    \dot x  = f(x,z),~~~  \epsilon \dot z  = g(x,z)
\end{equation}
where $\epsilon>0$ induces a separation between slow and fast dynamics. Under the assumption that, for each $x$, there exists a unique equilibrium $z^\star(x)$ satisfying $g(x,z^\star(x))=0$, the associated reduced model is 
$\dot x_{\mathrm{r}} = f(x_{\mathrm{r}}, z^\star(x_{\mathrm{r}}))$, obtained by letting $\epsilon\to 0^+$ in the fast subsystem. Stability and contractivity properties of such systems have been extensively studied; see, e.g.,~\cite{kokotovic1968singular,kokotovic1976singular,del2012contraction,khalil2002nonlinear}.

Our CT–DT interconnection~\eqref{eq:sample_data} admits a closely related interpretation. 
For fixed $x$, \eqref{eq:sample_data-discrete} has a unique equilibrium $z^*(x)$ satisfying $\mathsf{G}(x,z^*(x))=z^*(x)$ under Assumption \ref{as:contrDiscr-2}. 
In analogy with the classical CT--CT singular perturbation framework~\cite[Chapter~11]{khalil2002nonlinear}, one may view~\eqref{eq:sample_data-discrete} as a ``fast'' subsystem and~\eqref{eq:sample_data-plant} as its ``slow'' counterpart. However, unlike the CT--CT case, the ``reduced dynamics'' (i.e., \eqref{eq:rm-sampled}, obtained by letting $n\to +\infty$) need not be stable unless $T$ is sufficiently small. This motivates considering the limiting form of~\eqref{eq:rm-sampled}, obtained by letting $T \to 0^+$, which yields~\eqref{eq: reduced dynamics}.
 This viewpoint highlights how the pair $(T,n)$ jointly plays the role of the singular-perturbation parameter $\epsilon$ in~\eqref{eq:two-time-system-ct}.
\hfill$\Box$ 
\end{remark}

\subsection{Supporting results: a fundamental limitation, discrete-time contractivity, and continuous-time stability}

Before presenting the main technical results associated with the implication diagram in Fig.~\ref{fig:brief_diagram}, we highlight a fundamental limitation
of the CT–DT system: continuous-time contraction cannot be
expected under sampled-data implementations. This motivates a shift in perspective, whereby contraction properties are instead characterized in discrete time and combined with continuous-time stability arguments.

To begin with, we note that the interconnected system \eqref{eq:sample_data} is not contracting in
each open interval $(kT,(k+1)T)$. To see this, 
 any solution $y(t)$ to system \eqref{eq:sample_data},  $t \in (kT, (k+1)T)$, is also a solution of the following system: 
\begin{equation}\label{eq: constant-input system}
    \dot y=F(y):=\begin{bmatrix}
        f(x,z)\\~\mathbf{0}_{n_z}
    \end{bmatrix}.
\end{equation}
By \cite[Lemma~3.4]{FB:26-CTDS}, one has that $\osLip(F)\geq \mu(DF(y))$, where $DF(y)$ is the Jacobian of $F$. We note that $DF(y)$
has an eigenvalue $0$ and $ \mu(DF(y))\geq \alpha(DF(y))\geq 0$ holds by
\cite[Theorem~2.9(ii)]{FB:26-CTDS}. Hence, \eqref{eq: constant-input system} cannot be contracting.

In the following, we seek conditions for the interconnected system \eqref{eq:sample_data} to be $T$-DTC and GES as in equations~\eqref{eq: def_GES} and \eqref{eq: def_contraction_dt}. As a first step, we investigate the relationship between $T$-DTC and GES. To do so, we present a bound on the distance between trajectories within an interval $[kT, (k+1)T)$ and we show that $T$-DTC implies GES. 

\vspace{.1cm}

\begin{lemma}[Bound within $[kT, (k+1)T)$]
\label{lem:bound-x-by-xkT} Consider the interconnected system \eqref{eq:sample_data}, and let Assumptions~\ref{as: existence-of-weak-pairing} and \ref{as:Lipinterconnection-2} hold. Consider any two solutions $[\solxone(t)^\top,\solzone(t)^\top]^\top$ and $[\solxtwo(t)^\top,\solztwo(t)^\top]^\top$ to system \eqref{eq:sample_data}, and define $\tx(t) := \solxone(t) - \solxtwo(t)$ and $\tz(t) := \solzone(t) - \solztwo(t)$. For any $\tau \in [0, T)$, it holds that: 
\begin{equation}
\hspace{-.2cm} \begin{bmatrix}
    \| \tx(\tau + kT)\|_\mathcal{X}\\
        \left\| \tz(\tau + kT)\right\|_\mathcal{Z}
    \end{bmatrix}
    \leq  \mathcal{B}(T)  \begin{bmatrix}
    \left\| \tx(kT)\right\|_\mathcal{X}\\
        \left\| \tz(kT)\right\|_\mathcal{Z}
    \end{bmatrix}
\end{equation}
where the matrix $\mathcal{B}(T) \in \mathbb{R}^{2 \times 2}_{\geq 0}$ is defined by
\begin{align}
    \label{eq:def_B}
    \mathcal{B}(T) := 
    \begin{bmatrix}
       \max\{\e^{\xi T},1\} & \Lip_z(f) \int_0^T \e^{\xi(T-s)} ds \\ 0 & 1
    \end{bmatrix} \, ,
\end{align}
where $\xi < + \infty$ satisfies $\osLip_x(f)\leq\xi$.  \hfill $\Box$
\end{lemma}

\vspace{.1cm}

This result follows from Lemma~\ref{lemma:Input-state stability properties} presented in Section~\ref{sec:proofs}. In the rest of the paper, the integral of the form $\int_0^t \e^{c(t-s)} ds$ will be used frequently. To simplify the notation, we define the function $h:\mathbb{R}_{\geq 0}\times \mathbb{R}\to\mathbb{R}$ as
\begin{equation}\label{eq:def-of-h}
h(t,c):=\int_0^t \e^{c(t-s)} ds=\begin{cases}
    \frac{\e^{c t} -1}{c},~&c\neq 0\\
    t,~&c= 0
\end{cases}.
\end{equation}
We note that $h(t,c)$ is always non-negative and is equal to $0$ if and only if $t=0$. Moreover, for any $c$, $t \mapsto h(t,c)$ is strictly increasing since $\frac{\partial h(t,c)}{\partial t}= \e^{c t}$ or $1$. Furthermore, for any fixed $t>0$, $c \mapsto h(t,c)$ is also strictly increasing. Indeed, $\frac{\partial h(t,c)}{\partial c}=\frac{tc\e^{ct}-\e^{ct}+1}{c^2}=\frac{\e^{ct}}{c^2}(\e^{-ct}-(-ct+1))> 0$ if $c\neq 0$, and $\frac{\e^{c_1 t} -1}{c_1}\leq \frac{c_1 t}{c_1}=t=   \frac{c_2 t}{c_2}\leq\frac{\e^{c_2 t} -1}{c_2}$ for any $c_1<0$ and any $c_2>0$.

Leveraging Lemma \ref{lem:bound-x-by-xkT}, we show that $T$-DTC implies GES (and recall that $y = [x^\top, z^\top]^\top$). 

\vspace{.1cm}

\begin{proposition}[$T$-DTC implies GES]
\label{prop:contractivity-implies-GES}
Consider the system \eqref{eq:sample_data} and let Assumptions~\ref{as: existence-of-weak-pairing} and~\ref{as:Lipinterconnection-2} hold.  Let  $\| \cdot\|$ be a monotonic norm on $\mathbb{R}^2$, and let $\|y\|_{\textsf{cmp}}:=\| [\| x\|_\mathcal{X},\| z\|_\mathcal{Z} ]^\top\|$. Let $\xi < + \infty$ be such that $\osLip_x(f)\leq\xi$. If 
\begin{align*}
     \|  \solyone(kT) - \solytwo(kT) \|_{\textsf{cmp}} \leq b^k \| \solyone(0) - \solytwo(0) \|_{\textsf{cmp}}, \,\, k \in \mathbb{Z}_{\geq 0} ,  
\end{align*}
for some $b \in (0,1)$ and for any two solutions $\solyone(t), \solytwo(t)$, then the origin is GES for the interconnected system \eqref{eq:sample_data}. In particular,  it holds that
\begin{equation}
   \|y(t)\|_{\textsf{cmp}} \leq r\e^{-ct}  \|y(0)\|_{\textsf{cmp}}, \,\, t \geq 0 , 
\end{equation}
where $r:=\left\| \mathcal{B}(T)
    \right\| b^{-1}$ and $c:=-T^{-1}\ln b$. \hfill $\Box$
\end{proposition}

\vspace{.1cm}

The proof of this proposition is provided in Section~\ref{sec:proofs}. Proposition~\ref{prop:contractivity-implies-GES} asserts that, if the interconnected system \eqref{eq:sample_data} is $T$-DTC with rate $b$ with respect to a composite norm $\|\cdot\|_{\textsf{cmp}}$, then the system renders the origin GES with respect to the same composite norm. The result leverages the notion of composite norm; we refer the reader to~\cite[Ch.~2.7]{FB:26-CTDS} and~\cite{russo2012contraction} for a more in-depth discussion on  composite norms. 

With these preliminary results established, we now proceed to examine the conditions for contractivity and stability of the interconnected system \eqref{eq:sample_data}.

\subsection{Main results}
\label{sec: main results}

We now present the main technical results of the paper for the interconnected sampled-data system~\eqref{eq:sample_data}. The implication diagram in Fig.~\ref{fig:brief_diagram} is
reformulated in the more detailed and rigorous diagram in Fig.~\ref{fig:general_sys_diagram}, which summarizes the key quantitative
conditions and provides a structured overview of the logical relationships among the various sufficient conditions established in this section. The implication $\circled{2}\!\implies\!\circled{4}$ is immediate,
while the remaining implications are formalized by 
three corresponding theorems, presented next.

\begin{figure}[!ht]
    \centering

\noindent \framebox[.98\width]{
\begin{tikzpicture}
[>=stealth, ->, line width=1.5pt]
\node[align=center] (vc) at (-3.4,3.5) {\small{\textsf{Implication diagram for the general nonlinear system~\eqref{eq:sample_data}}} };
\node[align=center] (v1_1_c) at (-7.4,2.95) {\circled{1}};
\node[align=center] (v1_1_c) at (-2.8,2.95) {\circled{2}};
\node[align=center] (v2_1_c) at (-7.4,.7) {\circled{3}};
\node[align=center] (v2_2_c) at (-2.8,.7) {\circled{4}};
\node[align=center] (v2_1) at (-6,-.2) {\small{$\osLip_x(f(x, z^*(x)))< 0$} \\
\small{$\Lip_z(\sfG) < 1$}};
\node[align=center] (v2_2) at (-1.2,-.2) {\small{$\forall~n\in\mathbb{Z}_{>0},~\exists~T(n)\!>\!0$: }\\ \small{$\forall\,T\!<\!T(n)$, \eqref{eq:sample_data} is GES} 
};    
    \node[align=center] (v1_1) at (-6,2.0) {\small{$\osLip_x(f) < 0,  \Lip_z(\sfG)<1$} \\ \small{$-\osLip_x(f)(1-\Lip_z(\sfG))$}\\ \small{$>\Lip_z(f)\Lip_x(\sfG)$}};
\node[align=center] (v1_2) at (-1.2,2.0) {$\forall~n\in\mathbb{Z}_{>0}$ and $\forall~T\!>\!0:$ \\ 
\small{\eqref{eq:sample_data} is $T$-DTC and GES}
};
\draw[->]  (v1_1) edge  node[midway, above, sloped] {} (v1_2) ;
\draw[->] (v1_2) edge node[midway, right] {} (v2_2);
\draw[->]  (v2_1) edge node[midway, above, sloped] {} (v2_2);
\draw[->] (v1_1) edge node[midway, above, right] {} (v2_1);
\end{tikzpicture}
}
 \caption{Detailed implication diagram corresponding to Fig.~\ref{fig:brief_diagram} for the general nonlinear
CT-DT system~\eqref{eq:sample_data}. Under
Assumptions~\ref{as: existence-of-weak-pairing}--\ref{as:contrDiscr-2},
each arrow denotes a logical implication, where the
condition at the tail of the arrow is sufficient to guarantee the property
at the head.}
\vspace{-0.3cm}
    \label{fig:general_sys_diagram}
\end{figure}

\vspace{.1cm}

\subblue{\emph{1)}~\circled{1} $\implies$ \circled{2}: \emph{Contractivity and stability under small-gain conditions}:} Our analysis begins with the upper-left section of the implication diagram.

\vspace{.1cm}

\begin{theorem}[\circled{1} $\implies$ \circled{2}] \label{prop:Cont_Disc_11_to_12_general_case}
Consider the interconnected system \eqref{eq:sample_data}, and let Assumptions~\ref{as: existence-of-weak-pairing}, \ref{as:Lipinterconnection-2} and \ref{as:contrDiscr-2} hold with norms $\|\cdot\|_\mathcal{X}$ and $\|\cdot\|_\mathcal{Z}$. If  
\begin{equation}
\label{eq:small-gain}
-\osLip_x(f)(1-\Lip_z(\sfG))>\Lip_z(f)\Lip_x(\sfG),
\end{equation} 
then for any $n\in \mathbb{Z}_{>0}$ and for any $T>0$, there exists a weighted
vector norm $\|\cdot \|_{2,[\eta]}$ on $\mathbb{R}^2$ and a corresponding
composite norm $\|y\|_{\textsf{cmp}}:=\| [\| x\|_\mathcal{X},\|
  z\|_\mathcal{Z} ]^\top\|_{2,[\eta]}$ on $\mathbb{R}^{n_x+n_z}$ such that
the following holds:
\begin{itemize}
    \item[\emph{(i)}] The interconnected system \eqref{eq:sample_data} is $T$-DTC.  
    \item[\emph{(ii)}] The interconnected system \eqref{eq:sample_data} renders the origin GES. \hfill $\Box$
\end{itemize}
\end{theorem}

\vspace{.1cm}

The proof of the theorem is provided in Section~\ref{sec:proofs}.  We first note that  the condition~\eqref{eq:small-gain} imposes that $\osLip_x(f)<0$; this is because $\Lip_z(\sfG) <1$ and $\Lip_z(f)\Lip_x(\sfG) > 0$. This means that $x \mapsto f(x,z)$ must be strongly infinitesimally contracting, uniformly over $\mathcal{Z}$. We also note that~\eqref{eq:small-gain} is a general ``Lipschitz'' interconnection condition~\cite[Sec.~3.6]{FB:26-CTDS} and it can be interpreted as a ``small gain'' condition; in fact, it requires that the contraction rates of the sub-systems are larger than the interconnection gains.

Before proceeding with the next result, we provide four remarks. 

\vspace{.1cm}
\begin{remark}[No constraints on $T$ or $n$] 
\label{rem:No-constraints-T-n}
The small gain condition~\eqref{eq:small-gain} ensures stability of the interconnected system \eqref{eq:sample_data}, without imposing any condition on the sampling interval $T$ or the number of discrete-time steps $n$. We explain this (somewhat surprising) fact by noting that the claim \emph{(i)} of the Theorem~\ref{prop:Cont_Disc_11_to_12_general_case} is obtained by leveraging a \emph{gain matrix} $\mathcal{A}(n,T) \in \mathbb{R}^{2 \times 2}_{\geq 0}$ that allows us to derive a per-interval bound of the form: 
\begin{equation}
    \begin{bmatrix}
    \left\| x(kT)\right\|_\mathcal{X}\\
        \left\| z(kT)\right\|_\mathcal{Z}
    \end{bmatrix}\leq \mathcal{A}(n,T) \begin{bmatrix}
    \left\| x((k-1)T)\right\|_\mathcal{X}\\
        \left\| z((k-1)T)\right\|_\mathcal{Z}
    \end{bmatrix},
\end{equation}
which is valid for any $k = 1, 2, \ldots$. In particular, letting $\xi=\osLip_x(f)$ and $\mathcal{A}(n,T)=\begin{bmatrix}
    \mathcal{A}_{11} & \mathcal{A}_{12} \\
    \mathcal{A}_{21} & \mathcal{A}_{22}
\end{bmatrix}$, the entries are:    
\begin{align*}    
&\mathcal{A}_{11}:=        \e^{\xi T}, \hspace{.5cm} \mathcal{A}_{12}:=  \Lip_z(f)\frac{1-\e^{\xi T}}{-\xi},  \\
&\mathcal{A}_{21}:=
   \frac{1-\Lip_z(\sfG)^n}{1-\Lip_z(\sfG)}\Lip_x(\sfG) \e^{\xi T}\\
   &\mathcal{A}_{22}:= \frac{1-\Lip_z(\sfG)^n}{1-\Lip_z(\sfG)}\Lip_x(\sfG)  
 \Lip_z(f)\frac{1-\e^{\xi T}}{-\xi}+\Lip_z(\sfG)^n \, .
\end{align*}
Obviously, if $\mathcal{A}(n,T)$ is Schur stable, then,  the claim \emph{(i)} of Theorem~\ref{prop:Cont_Disc_11_to_12_general_case} follows (as we discuss in detail in the proof of the theorem). It turns out that the conditions for $\mathcal{A}(n,T)$   to be Schur stable do not depend on either $T$ or $n$. 
\hfill $\Box$
\end{remark}

\vspace{.1cm}

\begin{remark}[Comparison with Example~\ref{example:Low-dimensional-LTI}]
The small-gain condition~\eqref{eq:small-gain} (which implies $\osLip_x(f)<0$) rules out the scenario in Example \ref{example:Low-dimensional-LTI}, where the interconnected system becomes unstable when $T$ is large enough. This is because in Example~\ref{example:Low-dimensional-LTI}  we have that $\osLip_x(f) > 0$.  \hfill $\Box$
\end{remark}

\vspace{.1cm}

\begin{remark}[On the choice of the composite norm]
  In  Theorem~\ref{prop:Cont_Disc_11_to_12_general_case}, we use the weighted vector norm $\|\cdot \|_{2,[\eta]}$  in order to ensure that $\|\mathcal{A}(n,T) \|_{2,[\eta]}=\rho(\mathcal{A}(n,T))$ (where $\mathcal{A}(n,T)$ is defined in the previous remark). 
  To determine $\eta$, 
   one can use the Perron eigenvectors and the fact that the matrix $\mathcal{A}(n,T)$ is non-negative and irreducible; see \cite[Section~2.5]{FB:26-CTDS} and, in particular, \cite[Lemma 2.22]{FB:26-CTDS}. \hfill $\Box$
\end{remark}

\vspace{.1cm}

\begin{remark}[Detailed transient bound]
Theorem~\ref{prop:Cont_Disc_11_to_12_general_case}(ii) established that the  interconnected system \eqref{eq:sample_data} renders the origin GES. Leveraging Proposition~\ref{prop:contractivity-implies-GES}, we can characterize the transient bound as follows:
\begin{equation}
\label{eq:transient-bound-small-gain}
     \left\|   \begin{bmatrix}
    \left\| x(t)\right\|_\mathcal{X}\\
        \left\| z(t)\right\|_\mathcal{Z}
    \end{bmatrix}\right\|_{2,[\eta]}\leq r \e^{- a t}  \left\|  \begin{bmatrix}
    \left\| x(0)\right\|_\mathcal{X}\\
        \left\| z(0)\right\|_\mathcal{Z}
    \end{bmatrix}\right\|_{2,[\eta]}, \forall~t\geq 0,
\end{equation}
where $
r := \frac{1}{\rho\left( \mathcal{A}(n,T)\right)}\left\| \begin{bmatrix}
         1 &\frac{\e^{\xi T} -1}{\xi} \Lip_z(f)  \\ 0 & 1
    \end{bmatrix}
    \right\|_{2,[\eta]} 
$
and $a :=-\frac{1}{T}\ln (\rho\left( \mathcal{A}(n,T)\right))$.  \hfill $\Box$
\end{remark}

\vspace{.1cm}

\subblue{\emph{2)}~\circled{1} $\implies$ \circled{3}: \emph{small-gain conditions imply contractive RM}:}

The next main result
shows that the condition in Theorem \ref{prop:Cont_Disc_11_to_12_general_case} also implies contractivity of the RM.

\vspace{.1cm}

\begin{theorem}[\circled{1} $\implies$ \circled{3}]
\label{prop:Cont_Disc_11_to_21_general_case}
    Let Assumptions~\ref{as: existence-of-weak-pairing},
    \ref{as:Lipinterconnection-2} and \ref{as:contrDiscr-2} be satisfied
    with norms $\|\cdot\|_\mathcal{X}$ and $\|\cdot\|_\mathcal{Z}$. If
    \begin{equation}\label{eq:small-gain-condition}
      -\osLip_x(f)(1-\Lip_z(\sfG))>\Lip_z(f)\Lip_x(\sfG),
    \end{equation}
    then the RM \eqref{eq: reduced dynamics} is strongly infinitesimally
    contracting with respect to the norm $\|\cdot\|_\mathcal{X}$,
    i.e. $\osLip_x(f(x,z^*(x)))<0$. \hfill $\Box$
\end{theorem}

\vspace{.1cm}

The proof of the theorem is provided in Section~\ref{sec:proofs}. Theorem~\ref{prop:Cont_Disc_11_to_21_general_case} asserts that contractivity of the map $z \mapsto \sfG(x,z)$ and the small-gain condition~\eqref{eq:small-gain}, taken together, imply that the RM is strongly infinitesimally contracting with a given rate $\zeta$. The converse is, obviously, not true. This can be seen from the simple Example~\ref{example:Low-dimensional-LTI}; it is also at the foundations of feedback stabilization, where a feedback controller $z^*(x)$ is utilized to  stabilize a (possibly unstable) plant $\dot{x} = f(x,u)$. Lastly, we stress once again that the result of the theorem relies only on the Lipschitz interconnection gains, and does not involve conditions on $T$ or $n$.  The result of Theorem~\ref{prop:Cont_Disc_11_to_21_general_case}  prompts us to explore when contractivity of the RM may imply stability of the interconnected system~\eqref{eq:sample_data}.

\vspace{.1cm}

\subblue{\emph{3)}~\circled{3} $\implies$ \circled{4}: \emph{Stability from contractivity of RM}:} we consider the case where the RM~\eqref{eq: reduced dynamics} is designed to have contractivity properties. The main question is whether these properties can be preserved in the CT-DT system when $T > 0$ and $n < +\infty$. The main result is presented in the following.  

\vspace{.1cm}

\begin{theorem}[\circled{3} $\implies$ \circled{4}]\label{prop:Cont_Disc_21_to_22_general_case}
Consider the interconnected system~\eqref{eq:sample_data}, and let
Assumptions~\ref{as: existence-of-weak-pairing},
\ref{as:Lipinterconnection-2} and \ref{as:contrDiscr-2} hold with norms
$\|\cdot\|_\mathcal{X}$ and $\|\cdot\|_\mathcal{Z}$.  Let
$\xi\geq\osLip_x(f)$ and assume that $\zeta > 0$ is such that
$\osLip(f(x,z^*(x)) )\leq-\zeta$.  Then, for any $n\in \mathbb{Z}_{>0}$ there
exists a $T(n) > 0$ such that the interconnected system
\eqref{eq:sample_data} renders the origin GES for any $0<T<T(n)$. In
particular, such a $T(n)$ is given by:
\begin{equation}\label{eq:Tn}
\hspace{-.2cm}   T(n)=  \begin{cases}
   \displaystyle
 \frac{1}{\xi} \log\left(  \frac{\xi(1-[\Lip_z(\sfG)]^n)}{C_{2}(n)+C_{1}/\zeta}+1\right)  , &\text{ if $\xi\neq 0$ }\\[2.5ex]
  \displaystyle 
   \frac{1-[\Lip_z(\sfG)]^n}{C_{2}(n)+C_{1}/\zeta},&\text{ if $\xi=0$}
    \end{cases}    
    \hspace{-.2cm}
\end{equation}
where  $C_{1} > 0$ and $C_{2}(n) > 0$ are defined as
\begin{align}
\notag
   & C_{1}:= \frac{\Lip_z(f) \Lip_x(\sfG)}{1-\Lip_z(\sfG)}   \left(\Lip_x(f)+\frac{\Lip_z(f)\Lip_x(\sfG)}{1-\Lip_z(\sfG)}\right),\\
 &C_{2}(n):= [\Lip_z(\sfG)]^n \frac{\Lip_x(\sfG)\Lip_z(f)}{1-\Lip_z(\sfG)}. 
 \tag*{\text{$\Box$}}
\end{align}
\end{theorem}

\vspace{.2cm}

When the RM is strongly infinitesimally contracting,
Theorem~\ref{prop:Cont_Disc_21_to_22_general_case} asserts that, for any
$n\in \mathbb{Z}_{>0}$, stability of the interconnected
system~\eqref{eq:sample_data} can be preserved so long as $T <
T(n)$. Comparing Theorem~\ref{prop:Cont_Disc_21_to_22_general_case} with
the results of Theorem~\ref{prop:Cont_Disc_11_to_12_general_case}, one can
notice that Theorem~\ref{prop:Cont_Disc_11_to_12_general_case} does not
impose any conditions on $T$, but at the cost of stronger assumptions on
the interconnected system~\eqref{eq:sample_data}. In fact,
Theorem~\ref{prop:Cont_Disc_11_to_12_general_case} requires the
sub-system~\eqref{eq:sample_data-plant} to be strongly infinitesimally
contracting in $x$, uniformly in $z$. In view of
Theorem~\ref{prop:Cont_Disc_11_to_21_general_case}, this requirement is stricter
than assuming that the RM is contractive.

Inspecting the expression for $T(n)$ in~\eqref{eq:Tn}, one can observe the following: 
\begin{itemize}
    \item[a)] The bound $T(n)$ depends solely on Lipschitz constants of the maps $f$ and $\sfG$, and on the contractivity rate $\zeta$ of the RM.  
    \item[b)] For any $n\in\mathbb{Z}_{>0}$, $T(n)>0$. 
    \item[c)] If $\Lip_z(\mathsf{G}) > 0$ (corresponding to the case where $z^*$ can only be obtained via an iterative algorithm), $T(n)$ is strictly increasing with 
    respect to $n$. Therefore, the higher $n$, the less stringent are the requirements on $T$.
    \item[d)] $T(n)$ is monotonically decreasing with respect to both the minimal Lipschitz constant and the minimal one-sided Lipschitz constant. Consequently, the more conservative the estimates of these constants are, the smaller the value of $T(n)$ will be.
    
        \item[e)] For any $n\in\mathbb{Z}_{>0}$, $T(n)$ can be bounded as:
        \begin{equation}
        \label{eq:bound-on-T(n)}
        T(n)<\frac{\zeta}{\frac{\Lip_z(f) \Lip_x(\sfG)}{1-\Lip_z(\sfG)}\left(\Lip_x(f)+\frac{\Lip_z(f)\Lip_x(\sfG)}{1-\Lip_z(\sfG)}\right)}.
        \end{equation}
\end{itemize}

The last observation prompts the following remark. 

\vspace{.1cm}

\begin{remark}[Sampled-data implementation of $z^*(x)$]
Even when the map $z^*(x)$ can be computed (i.e., $n$ can be sufficiently large so that the Picard iterations converge to the unique fixed point),  there exists an upper bound on the  interval $T$ so that the system~\eqref{eq:rm-sampled} preserves the stability properties of its associated RM. For $n = \infty$, one has that: 
\begin{align}
\label{eq:Tn-infty}
    T(\infty)= \frac{1}{\xi} \log\left(  \frac{\xi \zeta}{C_{1}}+1\right).
\end{align} 
Note that \eqref{eq:bound-on-T(n)} is an upper bound for~\eqref{eq:Tn-infty}. Continuing the reference to optimization-based control methods, this means that the controller $z^*(x)$ needs to update the control input sufficiently fast in order to ensure stability of the closed-loop system~\eqref{eq:rm-sampled}. This translates into engineering specifications  in terms of allowable sensing and computational times.    
\hfill $\Box$ 
\end{remark}

\vspace{.1cm}

\begin{remark}[Detailed transient bound] 
\label{rem:transientbound}
We give a more detailed expression of the transient bound for the interconnected system~\eqref{eq:sample_data}. Define the matrix $\overline{\mathcal{A}}(n,T) :=\begin{bmatrix}
     \overline{\mathcal{A}}_{11} &  \overline{\mathcal{A}}_{12}\\
      \overline{\mathcal{A}}_{21} &  \overline{\mathcal{A}}_{22}
\end{bmatrix}$ where 
\begin{align*}
\overline{\mathcal{A}}_{11}:=&  \e^{-\zeta T}+    \frac{1-\e^{-\zeta T}}{\zeta}h(T,\xi) C_{1}, \\ 
     \overline{\mathcal{A}}_{12}:=&  \frac{1-\e^{-\zeta T}}{\zeta}(\Lip_z(f)+ h(T,\xi)C_{12}),\\
     \overline{\mathcal{A}}_{21}:=& h(T,\xi) C_{21}(n), \,\,\,   \overline{\mathcal{A}}_{22}:=h(T,\xi)C_{2}(n) +[\Lip_z(\sfG)]^n
\end{align*}
with $C_{1}, C_{2}(n)$ defined in the Theorem~\ref{prop:Cont_Disc_21_to_22_general_case} and with:
\begin{align*}
C_{12}:=&   \frac{\Lip_z(f)\Lip_x(\sfG)}{1-\Lip_z(\sfG)}\Lip_z(f) , \\
    C_{21}(n):=&  \frac{[\Lip_z(\sfG)]^n\Lip_x(\sfG)}{1-\Lip_z(\sfG)}\left( \Lip_x(f)+\frac{\Lip_x(\sfG)\Lip_z(f)}{1-\Lip_z(\sfG)} \right) .
\end{align*}
Note that matrix $\overline{\mathcal{A}}(n,T)$ is non-negative and irreducible. 
Therefore, by~\cite[Lemma~2.22]{FB:26-CTDS}, for any $n \in \mathbb{Z}_{>0}$ and $T > 0$ there exists  a weighted vector norm $\|\cdot \|_{2,[\bar{\eta}]}$  so that $\|\overline{\mathcal{A}}(n,T) \|_{2,[\bar{\eta}]}=\rho(\overline{\mathcal{A}}(n,T))$. Then, we have the following transient bound: 
\begin{equation}
\label{eq:transient-bound}
     \left\|   \begin{bmatrix}
    \left\| x(t)\right\|_\mathcal{X}\\
        \left\| z(t)\right\|_\mathcal{Z}
    \end{bmatrix}\right\|_{2,[\bar{\eta}]}\leq \varrho \e^{- \alpha t}  \left\|  \begin{bmatrix}
    \left\| x(0)\right\|_\mathcal{X}\\
        \left\| z(0)\right\|_\mathcal{Z}
    \end{bmatrix}\right\|_{2,[\bar{\eta}]}, \forall~t\geq 0,
\end{equation}
where 
$\varrho := \frac{1}{\rho\left( \overline{\mathcal{A}}(n,T)\right)}  \left\| \mathcal{B}(T)\right\|_{2,[\bar{\eta}]}
\left\| \begin{bmatrix}
        1 & 0 \\
        \frac{\Lip_x(\sfG)}{1-\Lip_z(\sfG)} & 1
    \end{bmatrix}\right\|_{2,[\bar{\eta}]}^2$
and $\alpha :=-\frac{1}{T}\ln \rho\left( \overline{\mathcal{A}}(n,T)\right)$, and where we recall that $\mathcal{B}(T)$ is defined in~\eqref{eq:def_B}. 
\hfill $\Box$ 
\end{remark}

\vspace{.1cm}

In the previous remark, we used a  weighted vector norm $\|\cdot \|_{2,[\eta]}$. This is further clarified in the following.

\vspace{.1cm}

\begin{remark}[Composite norm]
While omitted from the main statement of the theorem to improve readability, the exponential stability in Theorem~\ref{prop:Cont_Disc_21_to_22_general_case} is with respect to a composite norm $\|y\| = \| [\| x\|_\mathcal{X},\| z\|_\mathcal{Z} ]^\top\|_{2,[\eta]}$. Since the matrix $\overline{\mathcal{A}}(n,T)$ defined in Remark~\ref{rem:transientbound} is  non-negative and irreducible, one can find a vector $\eta = [\eta_1, \eta_2]^\top$ such that the weighted vector norm $\|\cdot \|_{2,[\eta]}$ satisfies that  $\|\overline{\mathcal{A}}(n,T) \|_{2,[\eta]}=\rho(\overline{\mathcal{A}}(n,T))$.
In particular, $\eta$ can be computed as in \cite[Lemma 2.22]{FB:26-CTDS}.
We note that, since  $\eta$ is computed using the Perron eigenvectors, it depends itself on $T$ and $n$. Therefore, $T$ and $n$ affect the bound~\eqref{eq:transient-bound} via $\varrho$ and $\alpha$, as well as the norm $\|\cdot \|_{2,[\eta]}$. 
\hfill $\Box$ 
\end{remark}

\vspace{.1cm}

In Theorem~\ref{prop:Cont_Disc_21_to_22_general_case}, we established GES of the origin under the condition that all assumptions hold on a forward-invariant set $\mathcal{X} \times \mathcal{Z}$. Our result can be generalized to cases where $\mathcal{X} \times \mathcal{Z}$ is not forward invariant, leading to a local version of Theorem~\ref{prop:Cont_Disc_21_to_22_general_case} as stated next.

\vspace{.1cm}

\begin{corollary}[Local exponential stability]\label{corollary:local_Cont_Disc_21_to_22_general_case}    
Consider the interconnected system~\eqref{eq:sample_data}, and let
Assumptions~\ref{as: existence-of-weak-pairing}-\ref{as:contrDiscr-2} hold with norms
$\|\cdot\|_\mathcal{X}$ and $\|\cdot\|_\mathcal{Z}$.  Let
$\xi\geq\osLip_x(f)$. Assume that $\zeta > 0$ is such that
$\osLip(f(x,z^*(x)) )\leq-\zeta$.  Given any $n \in \mathbb{Z}_{>0}$ and any $T$ with  $0 < T < T(n)$, where $T(n)$ is defined in \eqref{eq:Tn}, suppose that there exist subsets $\mathcal{X}_0 \subseteq \mathcal{X}$ and $\mathcal{Z}_0 \subseteq \mathcal{Z}$ such that every solution of \eqref{eq:sample_data} with initial condition in $\mathcal{X}_0 \times \mathcal{Z}_0$ remains in $\mathcal{X} \times \mathcal{Z}$ for all $t \geq 0$. Then, if $x(0) \in \mathcal{X}_0$ and $z(0) \in \mathcal{Z}_0$, the bound~\eqref{eq:transient-bound} holds.
 \hfill $\Box$
\end{corollary}

\subsection{Special case: Interconnected CT and DT LTI systems}
\label{sec: LTI-case}

Finally, we consider the special case of interconnected linear
time-invariant systems:
\begin{subequations}
\label{eq:sample_data_LTI}
\begin{align}
    \dot x(t) & = Ax(t)+Bz(t) \label{eq:sample_data-plant_LTI}\\
    z_{k} & = \sum_{i=0}^{n-1} D^i Cx(kT)+D^n z_{k-1},  \label{eq:sample_data-discrete_LTI} \\
   z(t) & = z_{k} , \,\, t \in [kT, (k+1)T), \label{eq:sample_data-z_LTI}
\end{align}
\end{subequations}
where $A \in \mathbb{R}^{n_x\times n_x}$, $B\in \mathbb{R}^{n_x\times
  n_z}$, $C \in \mathbb{R}^{n_z\times n_x}$, and $D \in
\mathbb{R}^{n_z\times n_z}$.  In other words, we let $\sfG(x,z) = C x + D
z$. Under Assumption~\ref{as:contrDiscr-2}, the matrix $D$ is Schur stable
and $z^*(x) = (I_{n_z}-D)^{-1}Cx$. Accordingly, the RM associated with the
interconnected LTI system~\eqref{eq:sample_data_LTI} is given by $
\dot
x_{\textsf{r}}(t) = \left(A+B(I_{n_z}-D)^{-1}C \right) x_{\textsf{r}}(t)$,
$t \geq 0$, with $x_{\textsf{r}}(t) \in \real^{n_x}$. In Fig.~\ref{fig:LTI_diagram}, we
present a customized version of the implication diagram for the
interconnected system~\eqref{eq:sample_data_LTI}.

\begin{figure}
    \centering
\vspace{.2cm}

\noindent \framebox[.98\width]{
\begin{tikzpicture}
[>=stealth, ->, line width=1.5pt]
\node[align=center] (vc) at (-3.6,3.5) {\small{\textsf{Implication diagram for the LTI system}~\eqref{eq:sample_data_LTI}} };
\node[align=center] (v1_1_c) at (-7.4,2.95) {\circled{1}};
\node[align=center] (v1_1_c) at (-2.5,2.95) {\circled{2}};
\node[align=center] (v2_1_c) at (-7.4,.7) {\circled{3}};
\node[align=center] (v2_2_c) at (-2.5,.8) {\circled{4}};
\node[align=center] (v2_1) at (-6,0) {\small{$A\!+\!B(I\!-\!D)^{-1}C$ is Hurwitz} \\
\small{$D$ is Schur}};
\node[align=center] (v2_2) at (-1.2,0) {\small{$\forall \, n\!\in\!\mathbb{Z}_{>0}~\exists~T(n)\!>\!0$: }\\ \small{$\forall\,T\!<\!T(n)$, \eqref{eq:sample_data_LTI} is} \\\small{  $T$-DTC and  GES} };    
    \node[align=center] (v1_1) at (-6,2.0) {\small{$\mu_{\|\cdot \|_\mathcal{X}}(A) < 0,  \|D \|_\mathcal{Z}<1$} \\ \small{$-\mu_{\|\cdot \|_\mathcal{X}}(A)(1-\| D\|_\mathcal{Z})$}\\ \small{$>\|B\|_{\mathcal{Z}\to\mathcal{X}} \| C\|_{\mathcal{X}\to\mathcal{Z}}$}};
\node[align=center] (v1_2) at (-1.2,2.0) {$\forall~n\in\mathbb{Z}_{>0}$ and $\forall~T\!>\!0$: \\ \small{\eqref{eq:sample_data_LTI} is $T$-DTC} and GES  };
\draw[->]  (v1_1) edge  node[midway, above, sloped] {} (v1_2) ;
\draw[->] (v1_2) edge node[midway, right] {} (v2_2);
\draw[->]  (v2_1) edge node[midway, above, sloped] {} (v2_2);
\draw[->] (v1_1) edge node[midway, above, right] {} (v2_1);
\end{tikzpicture}
}

 \caption{   
 Detailed implication diagram corresponding to Fig.~\ref{fig:brief_diagram} for 
 the CT-DT LTI system~\eqref{eq:sample_data_LTI}. In contrast to the general nonlinear case
(Fig.~\ref{fig:general_sys_diagram}), the claim in the bottom-right block additionally includes the $T$-DTC property.
 }
\vspace{-.5cm}
    \label{fig:LTI_diagram}
\end{figure}

Comparing the implication diagram for the interconnection of LTI systems~\eqref{eq:sample_data_LTI} with the one for~\eqref{eq:sample_data},  the main difference is in the implication $\circled{3} \implies \circled{4}$: in this case, if the RM is strongly infinitesimally contractive and the DT sub-system~\eqref{eq:sample_data-discrete_LTI} is contractive, then the interconnected system~\eqref{eq:sample_data_LTI} is $T$-DTC for any $T < T(n)$.  
This stronger result is due to the linearity of the discretization of system~\eqref{eq:sample_data-plant_LTI} and the linearity of $z^*(x)$; the result  is formalized in the following.  

\vspace{.1cm}

\begin{proposition}[\circled{3} $\implies$ \circled{4} for~\eqref{eq:sample_data_LTI}]
Consider the interconnected system~\eqref{eq:sample_data_LTI}. Assume that the matrix 
$A\!+\!B(I_{n_z}\!-\!D)^{-1}C$ is Hurwitz and that $D$ is Schur. Then, for any
$n\in\mathbb{Z}_{>0}$ and $T<T(n)$ (with $T(n)$ as in~\eqref{eq:Tn}), the
following holds for any two solutions $\solyone(t)$ and $\solytwo(t)$ of
\eqref{eq:sample_data_LTI}:
\begin{align}
    \label{eq:dtc-lti}
    \|\solyone(kT) \! - \! \solytwo(kT)\| \leq \left(\frac{1\!+\!\rho(\mathcal{L}(n,T))}{2}\right)^k\|\solyone(0) \! - \! \solytwo(0)\| 
\end{align}
for any $k \in \mathbb{Z}_{\geq 0}$, where $\mathcal{L}(n,T) \in \mathbb{R}^{(n_x+n_z) \times (n_x+n_z)}$ is Schur stable, with block entries:   $\mathcal{L}_{11}:=\e^{AT}$,  $\mathcal{L}_{12}:= \left(\int_{0}^T \e^{A(T-s)}~ds \right)B$,    
$\mathcal{L}_{21}:=(I_{n_z}\!-\!D^n)(I_{n_z}\!-D)^{-1}C \e^{AT}$,  $    \mathcal{L}_{22}:=    (I_{n_z}\!-D^n)(I_{n_z}\!-\!D)^{-1}C\mathcal{L}_{12}+\!D^n$. 
Moreover, for any $n\in\mathbb{Z}_{>0}$ and $T<T(n)$,~\eqref{eq:sample_data_LTI} renders the origin GES. 
\hfill $\Box$
\end{proposition}

\vspace{.1cm}

The bound~\eqref{eq:dtc-lti} establishes that the interconnected system~\eqref{eq:sample_data_LTI} is $T$-DTC. We omit the proof to avoid repetitive arguments and save space. However, as a sketch, GES follows from Theorem \ref{prop:Cont_Disc_21_to_22_general_case}, which also implies in this specific setup that $\mathcal{L}(n,T)$ is Schur stable with any $n\in\mathbb{Z}_{>0}$ and $T<T(n)$. Finally, for any $n\in\mathbb{Z}_{>0}$ and $T<T(n)$, by \cite[Lemma 2.17]{FB:26-CTDS}, there exists a norm $\| \cdot\|$ on $\mathbb{R}^{n_x+n_z}$ such that the induced matrix norm  $\|\mathcal{L}(n,T) \|<\frac{1+\rho(\mathcal{L}(n,T))}{2}$.

\section{Application to Model Predictive Control}
\label{sec:mpc}

As previously discussed, a key motivation for our analysis is optimization-based control. In this setting, $\dot{x}(t) = f(x(t), z(t))$ as  in~\eqref{eq:sample_data-plant} is a given \emph{plant} and:

\noindent $\bullet$ $z^*(x(t))$ is a state feedback law that is implicitly defined as the optimal solution of a parametric optimization problem -- where the plant state $x(t)$ serves as the problem's parameter.

\noindent $\bullet$ $\sfG(x,z)$ is the map of an optimization algorithm utilized to compute $z^*(x)$ via $\lim_{n \rightarrow \infty}\sfGn(x,z_0) = z^*(x)$; and,

Using the interconnected system~\eqref{eq:sample_data} in lieu of the RM~\eqref{eq: reduced dynamics} represents a case where the control input $z(t)$ is produced in a sampled-data and online fashion, i.e., based on an early termination of the optimization algorithm.   
Our implication diagram suggests two approaches to designing a  controller:

\vspace{.1cm}

\noindent $\bullet$ $\circled{3}\!\implies\!\circled{4}$ \emph{(RM-based design)}: Design the controller $z^*(x)$ so that the closed-loop system $\dot{x}(t) = f(x(t), z^*(x(t)))$ is contractive, and then utilize Theorem~\ref{prop:Cont_Disc_21_to_22_general_case} or Corollary~\ref{corollary:local_Cont_Disc_21_to_22_general_case} to impose constraints on $n$ and $T$. 

\vspace{.1cm}

\noindent $\bullet$ $\circled{1} \!\implies\!\circled{2}$ \emph{(Small-gain-based design)}: Under the stronger assumption that the plant $\dot{x}(t) = f(x(t), z(t))$ is already contractive (i.e., $\osLip_x(f) < 0$), design the controller using Theorem~\ref{prop:Cont_Disc_11_to_12_general_case}. 

\vspace{.1cm}

Obviously, the second approach implies that the RM $\dot{x}(t) = f(x(t), z^*(x(t)))$ is contractive in view of the implication $\circled{1} \! \implies \! \circled{3}$. In the following, we illustrate a specific case where $z^*(x)$ is given by an MPC strategy~\cite{rawlings2017model,gros2020linear}.

\subsection{Batch and online MPC}

Consider an LTI system $\dot x(t)=Ax(t)+ B u(t)$, where $A \in \mathbb{R}^{n_x\times n_x}$ and $B\in \mathbb{R}^{n_x\times n_u}$. To formulate an instance of MPC for this system, let  $\Delta > 0$  be an interval used to discretize the dynamics of the LTI system, and let $H$ be the length of the MPC  horizon (here, $\Delta > 0$ is an MPC design parameter that is used to build the predictive horizon). Define the matrices: 
\begin{align}
\bar{A}:= \e^{A\Delta},~~~\bar{B}:=\left(\int_{0}^\Delta \e^{A(\Delta-s)}~ds\right) B \, .
\end{align}

At a given time $kT$, instances of the MPC problem can then be built based on the discretized dynamics 
$x_{k,i+1} = \bar{A} x_{k,i} +\bar{B} z_{k,i}$, $i = 1, \ldots, H$, with $x_{k,1} = x(kT)$. We first consider the following unconstrained MPC formulation~\cite{morari1988model} (and comment on constrained versions in Section~\ref{sec:conclusions}) that is solved at each time $t=k T$, $k \in \mathbb{Z}_{> 0}$
:
\vspace{-0.1cm}
\begin{align}
z^*(x(kT)) \! = \! 
  & \argmin_{\substack{z \in \mathbb{R}^{n_u H} \\ x \in \mathbb{R}^{n_x (H\!+\!1)}}}  
 \sum_{i = 1}^{H} \!\left(U_i(z_i) \!+\! J_i(x_i) \right) \!+\! J_{H+1}(x_{H+1}) \nonumber \\
& \text{s. to:~} x_{i+1} = \bar{A} x_{i} +\bar{B} z_{i}, ~~ i = 1, \ldots, H\! \nonumber \\
& ~~~~~~~ x_{1} =  x(kT) , \label{eq:mpc-unconstrained} 
\end{align} 
where $z_i \mapsto U_i(z_i)$ is a cost function associated with the input $z_i\in\mathbb{R}^{n_u}$ and $x_i \mapsto J_i(x_i)$ is a cost function associated with the state of the system at time $kT + (i-1) \Delta $. We note that the functions $U_i(z_i)$ and $J_i(x_i)$ can also include (soft) barrier functions associated with input and state constraints, respectively. Once~\eqref{eq:mpc-unconstrained} is solved (i.e., a \emph{batch} solution), a sampled-data implementation of the MPC amounts to:
\begin{subequations}
\label{eq:mpc-online-optimal} 
\begin{align} 
\dot x(t)&=Ax(t)+ B \Pi_1 z(t) , \\
    z_{k} & = z^*(x(k T)), \\
   z(t) & = z_{k} , \,\, t \in [kT, (k+1)T), 
\end{align}
\end{subequations}
where $\Pi_1 = [\mathbf{I}_{n_u \times n_u}, \mathbf{0}_{n_u \times n_u (H-1)}]$  is  a matrix  that extracts the block $z_1$ from $z^*(x(k T))$. Here, at each time instant $k T$, the MPC problem is solved over the predictive interval $kT, kT+\Delta, \ldots, kT+ H\Delta$, and the input $z_1^*(x(k T)) = \Pi_1 z^*(x(k T))$ is applied over the interval $[kT, (k+1)T)$~\cite{morari1988model}.   

We consider now the ideal case where the MPC problem can be solved continuously in time (i.e., for $T \rightarrow 0^+$)~\cite{figura2024instant}; that is, given the state $x(t)$, $t \in \mathbb{R}$, the following problem is solved instantaneously  over a rolling predictive horizon $[t, t+H\Delta]$: 
\begin{align}
z^*(x(t)) \! = \! & \argmin_{\substack{z \in \mathbb{R}^{n_u H} \\ x \in \mathbb{R}^{n_x (H\!+\!1)}}}   \sum_{i = 1}^{H} \!\left(U_i(z_i) \!+\! J_i(x_i) \right) \!+\! J_{H+1}(x_{H+1}) \nonumber \\
& \text{s. to:~} x_{i+1} = \bar{A} x_{i} +\bar{B} z_{i}, ~~ i =1, 2, \ldots, H\! \nonumber \\
& ~~~~~~~ x_{1} =  x(t), \label{eq:mpc-unconstrained-continuous}  
\end{align}
where the input to the MPC problem is now $x(t)$. In this case, one has  the closed-loop system:
\begin{align}
\label{eq:mpc-solu}  
\dot x(t)=Ax(t)+ B \Pi_1 z^*(x(t)) \, , ~~ t \geq 0 
\end{align}
for $x(0) \in \mathbb{R}^{n_x}$. The system~\eqref{eq:mpc-solu} corresponds to the RM in our framework. 

Note that the optimization problem~\eqref{eq:mpc-unconstrained} can be reduced to the following form:
\begin{align}
\label{eq:mpc-unconstrained2} 
z^*(x(kT)) \! = \! \argmin_{z \in \mathbb{R}^{n_u H}}  \tilde{U} (z) + \tilde{J}(z,x(kT))  
\end{align}
where we define $\tilde{U}(z) := \sum_{i = 1}^{H}  U_i(z_i)$, $\tilde{J}(z,x(kT)) := \sum_{i = 1}^{H+1} \tilde{J}_i(z, x(kT))$ and each of the functions $\tilde{J}_i(z, x(kT))$ is obtained from $J_i(x_i)$ by  repeatedly substituting the dynamics for the state; see, for example,~\cite[Section~14.2]{gros2020linear}. Assuming that the cost in~\eqref{eq:mpc-unconstrained2} is $\mu$-strongly convex and $\ell$-smooth in $z$ uniformly in $x$, one can utilize the following gradient descent map:
\begin{align}
\label{eq:mpc-G} 
\sfG_{\textsf{mpc}}(x,z) := z - \eta \left( \nabla_z \tilde{U}(z) +  \nabla_z \tilde{J}(z; x) \right)
\end{align}
where $\eta \in (0, 2 \mu/\ell^2)$ is the step size. With this choice of $\eta$, the map $z \mapsto \sfG(x,z)$ is contractive with respect to $\|\cdot\|_2$. Then, a sampled-data and \emph{sub-optimal} implementation of the MPC is: 
\begin{subequations}
\label{eq:mpc-online} 
\begin{align} 
\dot x(t)&=Ax(t)+ B \Pi_1 z(t) , \\
    z_{k} & = \sfG_{\textsf{mpc}}^{n}(x(kT),z_{k-1}), \\
   z(t) & = z_{k} , \,\, t \in [kT, (k+1)T) .
\end{align}
\end{subequations}
When the MPC is designed so that the closed-loop system~\eqref{eq:mpc-solu} is strongly infinitesimally contractive, the results of Theorem~\ref{prop:Cont_Disc_21_to_22_general_case} and Corollary~\ref{corollary:local_Cont_Disc_21_to_22_general_case} can be readily applied to the sub-optimal MPC~\eqref{eq:mpc-online} to identify requirements on $n$ and $T$ in order to preserve stability. 

Interestingly, an important aspect of  Theorem~\ref{prop:Cont_Disc_21_to_22_general_case} in the context of sub-optimal MPC is that stability can be preserved even when $n = 1$, so long as $T$ is chosen sufficiently small. 
This guarantee relies on the contractivity of the reduced model induced by
the MPC controller, which is conceptually different from and not directly comparable to the classical
assumptions adopted in time-distributed or online MPC frameworks; see,
e.g.,~\cite{liao2020time,liao2021analysis}. 
Moreover, existing works typically require $n>1$ and therefore do not
address the single-iteration case considered here.

\vspace{-0.2cm}
\subsection{Numerical examples}
In this subsection, we consider a double integrator with 
$A=\begin{bmatrix}
    0 & 1 \\ 0 & 0
\end{bmatrix}$, and 
$B=\begin{bmatrix}
    0\\1
\end{bmatrix}$. Then given $\Delta>0$, one has that  $\bar{A}=\begin{bmatrix}
    1 & \Delta \\ 0 & 1
\end{bmatrix}$, $\bar{B}=\begin{bmatrix}
    0.5\Delta^2\\ \Delta
\end{bmatrix}$. Next, we set $\Delta=0.2$ and  consider the MPC with a predictive horizon $H=5$. 

\subsubsection{First experiment} In the first experiment, similar to, e.g.,~\cite{morari1988model} and~\cite[Sec. 2.5.1]{rawlings2017model}, we consider the cost functions $U_i(z_i):=\|z_i\|_R^2=|z_i|^2$ and $J_i(x_i):=\| x_i\|_Q^2=\|x_i \|^2$, $1\leq i\leq H$, where we set $R=1$ and $Q=\mathbf{I}_2$. The terminal cost $J_{H+1}(x)$ is set as $\|x \|_P^2$, where $P$ is the solution to the discrete-time algebraic Riccati equation:
$$
\bar{A}^\top P \bar{A}-(\bar{A}^\top P \bar{B})(\bar{B}^\top P \bar{B}+R)^{-1}(\bar{A}^\top P \bar{B})^\top+Q=0.
$$
In this case, $\mathcal{X}=\mathbb{R}^2$, $\mathcal{Z}=\mathbb{R}^H$, and the mapping $z^*$ in \eqref{eq:mpc-unconstrained} is a linear mapping $z^*(x)=-K_{\textsf{uncon}}x$~\cite[Sec. 2.5.1]{rawlings2017model}.

Numerically, the resulting closed-loop matrix for the RM is 
$$A_\textsf{cl}:=A-B\Pi_1 K_{\textsf{uncon}}=\begin{bmatrix}
             0  &   1 \\
    -0.8412 &  -1.5460
\end{bmatrix},$$ 
 with $\mu_{2,P}(A_\textsf{cl})=\frac{1}{2}\lambda_{\max}\left(PA_\textsf{cl}P^{-1}+A_\textsf{cl}^\top \right)=-0.4407$  and $\alpha(A_\textsf{cl})=-0.7730$. This implies that  $A_\textsf{cl}$ is Hurwitz and the RM is contractive with respect to the norm $\|\cdot\|_{2,P}$ over $\mathbb{R}^2$. We then define the mapping $\sfG_{\textsf{mpc}}(x, z)$ as the gradient descent algorithm introduced in \eqref{eq:mpc-G}. Under this setup, Assumptions \ref{as: existence-of-weak-pairing}-\ref{as:contrDiscr-2} are satisfied and the sampled-data system with any $(n,T)$ renders the set $\mathcal{X}\times\mathcal{Z}$  forward invariant. Therefore, Theorem~\ref{prop:Cont_Disc_21_to_22_general_case} is verified in this scenario.

\begin{figure}
    \centering
    \includegraphics[width=.9\linewidth]{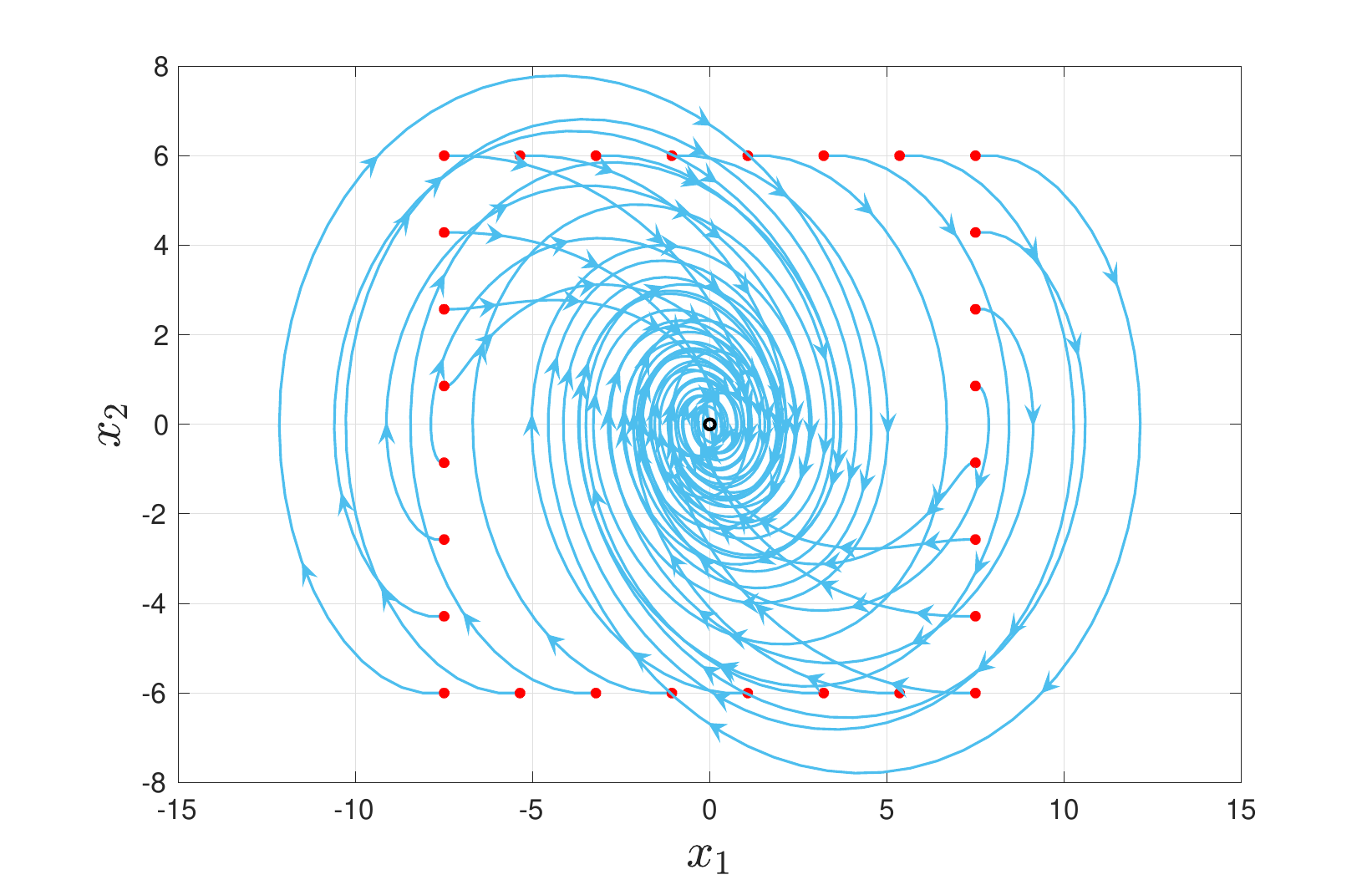}
    \includegraphics[width=.9\linewidth]{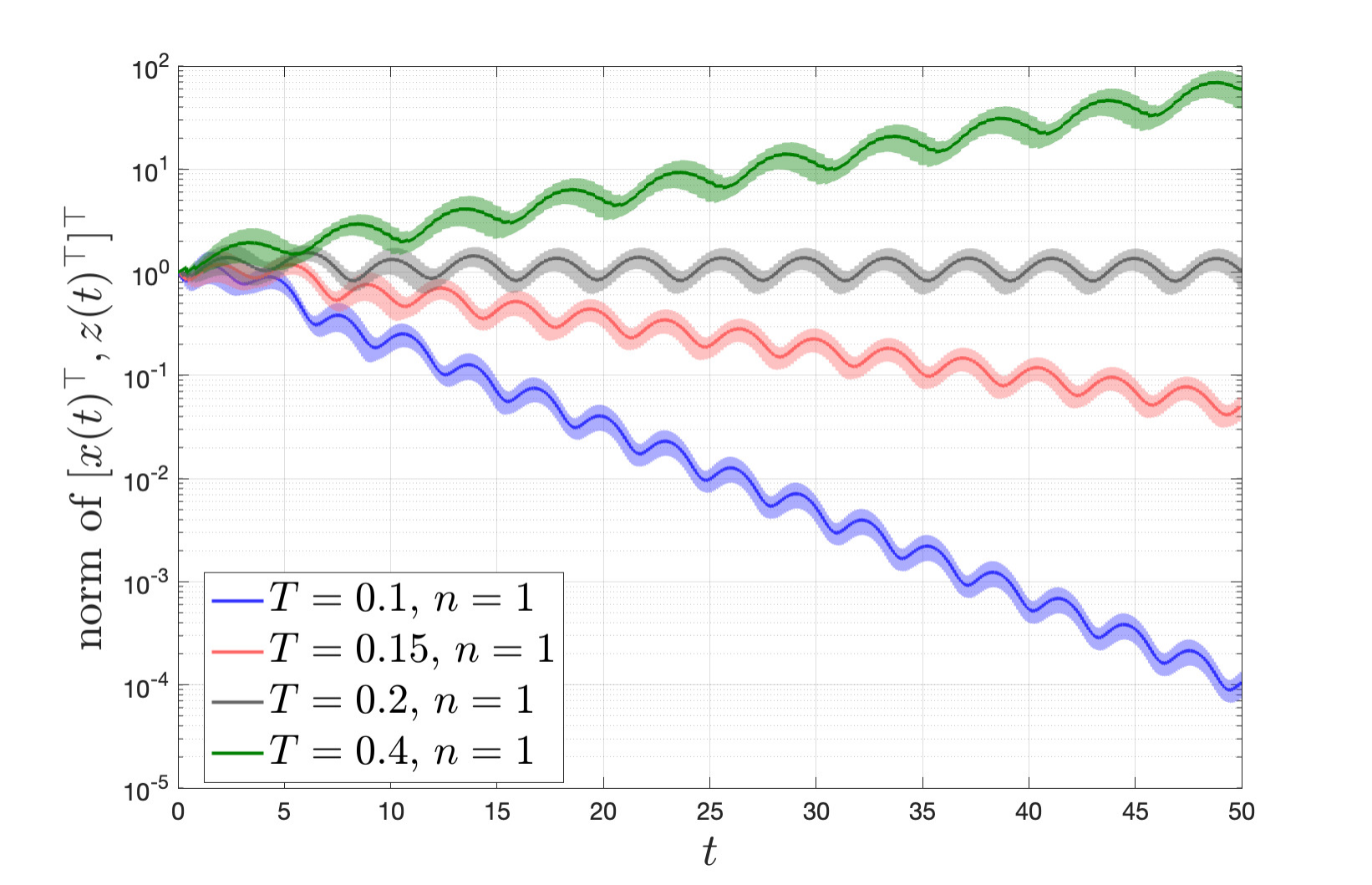}
    \caption{\emph{(Top)} Phase portrait of the suboptimal MPC  when $n = 1$ and  $T=0.1$; the initial conditions of $x$ correspond to the red points, and the initial condition of $z$ is always $[0,0,0,0,0]^\top$. A similar behavior was obtained when considering different initial conditions.   \emph{(Bottom)} Composite norm of $x(t), z(t)$; the plot shows the mean and the confidence band computed from $100$ randomly picked, i.i.d. initial conditions $x(0),z(0)$. It can be seen that the sub-optimal MPC~\eqref{eq:mpc-online} is stable even for $n = 1$ when $T$ is sufficiently small; for larger values of $T$, the closed-loop system~\eqref{eq:mpc-online} becomes unstable. }
    \vspace{-0.2cm}
    \label{fig:uncon-MPC}
\end{figure}
 
In Fig.~\ref{fig:uncon-MPC}, the top plot (phase portrait) illustrates several trajectories of the state $x(t)$ for different initial conditions; here, we set $n = 1$ and $T = 0.1$, with $T$ verifying the conditions of Theorem~\ref{prop:Cont_Disc_21_to_22_general_case}.  The bottom plot shows the $\ell_2$ norm of $[x(t)^\top, z(t)^\top]$ over time, along with a confidence band computed from $100$  independent and identically distributed (i.i.d.) initial conditions. Each initial condition is sampled from the uniform distribution in the interval (0,1), and normalized to have unit $\ell_2$ norm. With $n$ fixed at 1, the system is exponentially stable for small enough values of $T$ and then becomes unstable. This behavior supports our analysis in Theorem~\ref{prop:Cont_Disc_21_to_22_general_case}, showing that a sufficiently small $T$ is necessary to ensure stability when $n = 1$.

\subsubsection{Second experiment} We consider a case where ``soft constraints'' on the state $x(t)$ are added to the cost of the MPC problem~\eqref{eq:mpc-unconstrained}; in this case, the optimal policy $z^*(x)$ can no longer be computed in closed form.  
In particular, we consider the following formulation: 
\begin{align}
z^*(x(t)) \! = \! &\argmin_{\substack{z \in \mathbb{R}^{n_u H} \\ x \in \mathbb{R}^{n_x (H\!+\!1)}}}  \sum_{i = 1}^{H} \!\left( \|x_{i} \|_Q^2 +\| z_i\|_R^2   \right) \!+\! \|x_{H+1}\|_P^2 \nonumber \\
&~~~~~~~~~ +\gamma\sum_{i=1}^{H+1} \left\|\max\left\{\begin{bmatrix}
    0 \\ 0
\end{bmatrix},x_i-\begin{bmatrix}
    10\\3
\end{bmatrix}\right\}\right\|^2 \nonumber \\
&~~~~~~~~~ +\gamma\sum_{i=1}^{H+1} \left\|\max\left\{\begin{bmatrix}
    0 \\ 0
\end{bmatrix},\begin{bmatrix}
    -10\\-3
\end{bmatrix}-x_i\right\}\right\|^2 \nonumber \\
& \text{s. to:~} x_{i+1} = \bar{A} x_{i} +\bar{B} z_{i}, ~~ i = 1, \ldots, H\! \nonumber \\
& ~~~~~~~ x_1 = x(t),  
\label{eq:soft-constrained MPC}
\end{align}
where  $\gamma>0$,  the $\max$ operator is applied entry-wise, and the rest of the parameters are the same as in the unconstrained case.  Intuitively, the soft barrier penalizes violations of constraints of the form 
$[-10, -3]^\top \leq x_i\leq [10, 3]^\top$, $i=1,2,..,H+1$. 

When $\gamma$ is small, we expect a behavior similar to the one analyzed previously (which one recovers by setting $\gamma = 0$), where the RM is globally contractive; on the other hand, an interesting behavior is noted numerically with the increase of $\gamma$:  when 
$\gamma$ is sufficiently large,~\eqref{eq:soft-constrained MPC} approximates solutions of an MPC with hard constraints on the state $x$, and the contractivity of the associated RM is observed not to be global. This is evident from Fig.~\ref{fig:soft-MPC}, where we show the contour plot of $\mu_{2,P}(A+B\Pi_1 J_{z^*}(x))$, which is the log-norm of the Jacobian of the reduced model, obtained numerically by
applying central difference with stepsize $0.01$. The figure considers three cases: $\gamma = 1$ (left), $\gamma = 10$ (center), and $\gamma = 100$ (right). 

\begin{figure}[h]
    \vspace{-0.2cm}
    \centering
    \includegraphics[width=1.0\linewidth]{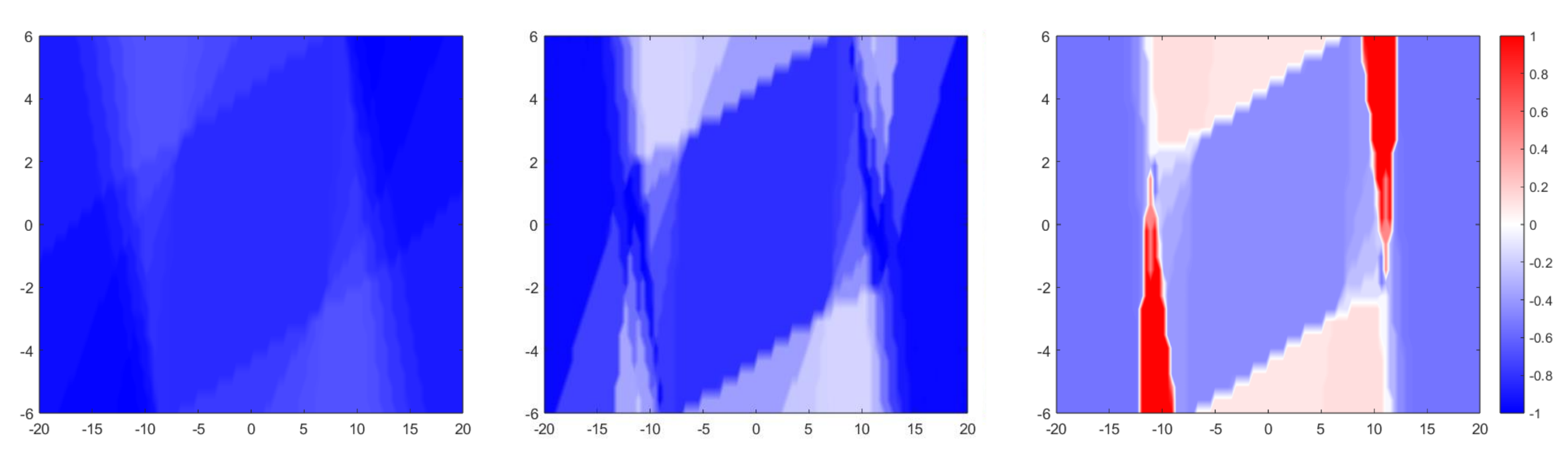}
    \caption{Contour plot of $\mu_{2,P}(A+B\Pi_1 J_{z^*}(x))$, where $z^*(x)$ is the optimal solution map of~\eqref{eq:soft-constrained MPC}; this represents an estimate of the one-sided Lipschitz constant of the RM, for different values of $\gamma$ in~\eqref{eq:soft-constrained MPC}. The blue region signifies contractivity of the RM. The plot shows the cases  $\gamma = 1$ (left), $\gamma = 10$ (center), and $\gamma = 100$ (right).}
    \vspace{-0.2cm}
    \label{fig:soft-MPC}
\end{figure}

In this case, the numerical results presented are used to validate the Corollary~\ref{corollary:local_Cont_Disc_21_to_22_general_case}, since contractivity of the RM is not global. In the following numerical experiment, we fix the penalty rate as $\gamma=10$. Moreover, we set  $\mathcal{X}_0= [-10,10]\times [-3,3]$ and $\mathcal{X}=[-20,20]\times [-6,6]$, based on the contour plot in Fig.~\ref{fig:soft-MPC}(center). We note that the cost function in \eqref{eq:soft-constrained MPC} is strongly convex and smooth since the Hessian of the penalty term satisfies that $0_{2\times2 }\preceq  \text{Hessian}\preceq2 \gamma \mathbf{I}_{2\times2 }$. Hence, we again utilize the gradient descent defined in \eqref{eq:mpc-G} as $\sfG_{\textsf{mpc}}(x,z)$. Since the unique equilibrium of the RM is $[0,0]^\top$, we can simply set $\mathcal{Z}_0$ (the set of initial conditions for $z$) as $\mathcal{Z}_0=\{ 
 z^*(0) \}=\{\mathbf{0}_H\}$ and set $\mathcal{Z}=\mathbb{R}^H$. Under this setup, we ensure the satisfaction of Assumptions \ref{as: existence-of-weak-pairing}-\ref{as:contrDiscr-2}.

\begin{figure}[t]
    \centering
\includegraphics[width=.9\linewidth]{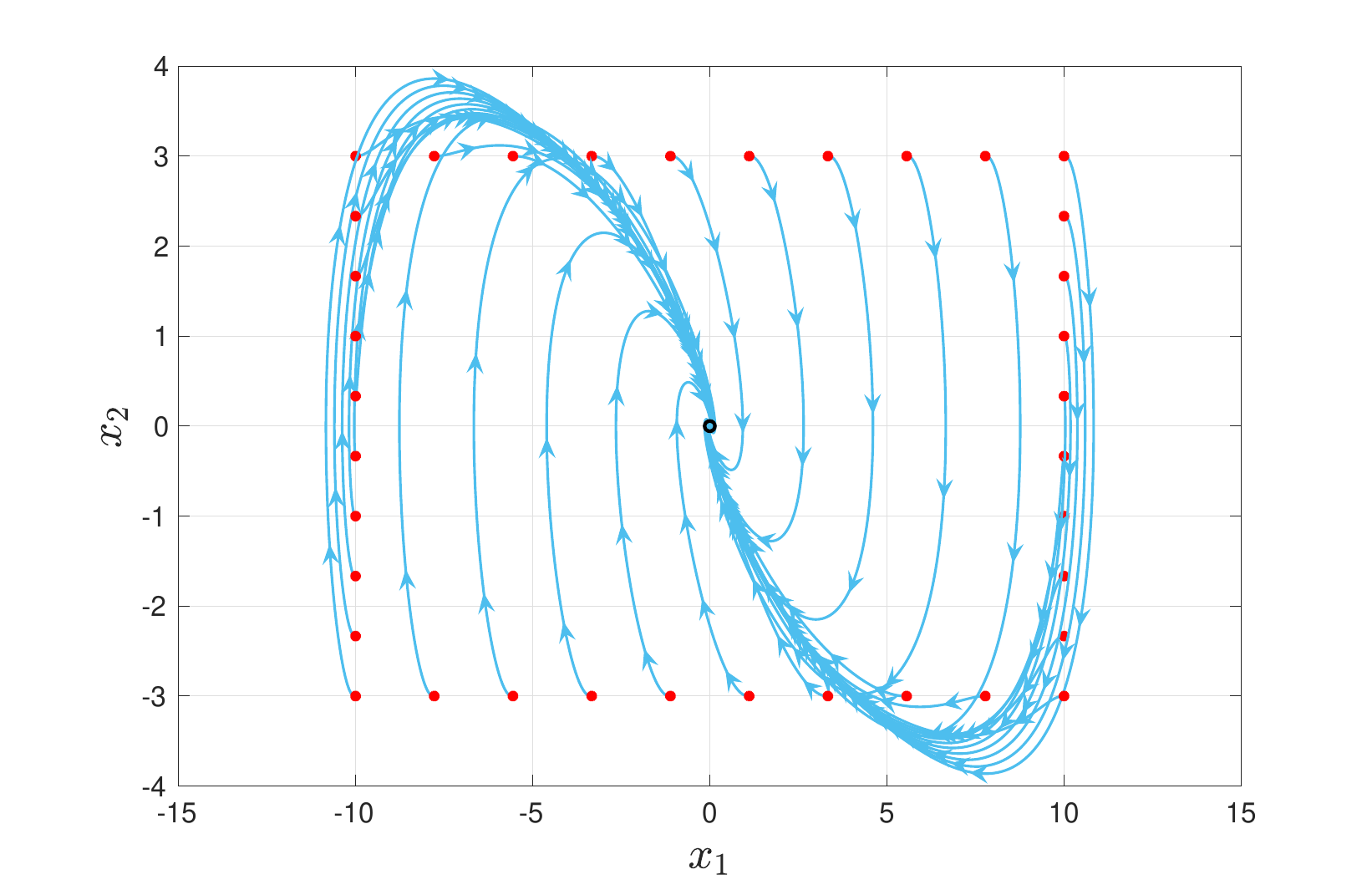}\\    
\includegraphics[width=.9\linewidth]{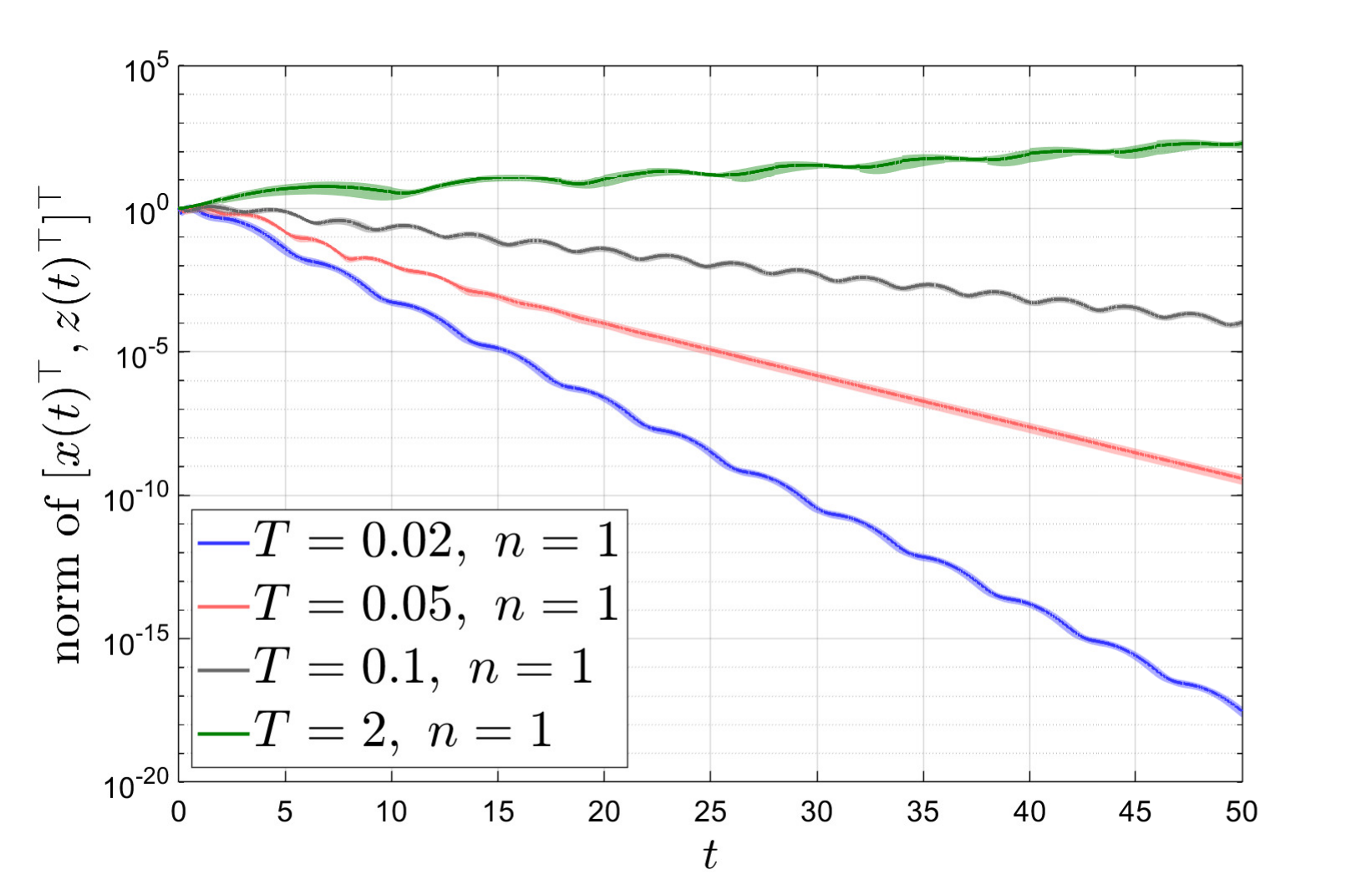}
    \caption{\emph{(Top)} Phase portrait of the suboptimal MPC  when $n = 1$ and  $T=0.02$; the initial conditions of $x$ correspond to the red points, and the initial condition of $z$ is always the origin.   \emph{(Bottom)} Composite norm of $x(t), z(t)$; the plot shows the mean and the confidence band computed from $100$ randomly picked, i.i.d. initial conditions $x(0),z(0)$.  }
    \label{fig:soft-MPC, single iteration}
    \vspace{-0.3cm}
\end{figure}

Next, for any given $(n,T)$ and any initial condition
$(x_0,z_0)\in\mathcal{X}_0\times\mathcal{Z}_0$, we need to verify that the
solution $(x(t),z(t))$ does not exit the set
$\mathcal{X}\times\mathcal{Z}$.  In the top plot of
Fig.~\ref{fig:soft-MPC, single iteration}, we verify this condition with
$(n,T)=(1,0.02)$. The phase portrait is generated with initial condition
$x$ on the boundary of $\mathcal{X}_0$ and initial condition
$z=\mathbf{0}_H$. One can observe that all the solutions stay in
$\mathcal{X}$.  The bottom plot of Fig.~\ref{fig:soft-MPC, single iteration} certifies that when $n = 1$ and $T=0.02$, the origin is
exponentially stable with any initial conditions $[x_0^\top,z_0^\top]^\top
\in \mathcal{X}_0\times\mathcal{Z}_0.$ Similar to the previous example,
increasing $T$ eventually leads to instability.

As a future step, we will investigate the application of Corollary~\ref{corollary:local_Cont_Disc_21_to_22_general_case} for the case of an MPC with hard constraints on both the inputs $z$ and the state $x$, as highlighted in Section~\ref{sec:conclusions}.

\section{Proofs}
\label{sec:proofs}
\subsection{Useful lemmas}

In this subsection, we provide some results that will be leveraged in the subsequent development.

\vspace{.1cm}

\begin{lemma}[${\sfGn}(x,z)$ is Lipschitz]
\label{lemma: Lip-of-gn}
Consider the map ${\sfG}(x,z)$ and let
Assumptions~\ref{as:Lipinterconnection-2} and~\ref{as:contrDiscr-2}
hold. Then, the map ${\sfGn}(x,z)$ is Lipschitz in both arguments and the
associated (minimal) Lipschitz constants $\Lip_x(\sfGn)$ and
$\Lip_z(\sfGn)$ satisfy:
\begin{subequations}\label{eq:g_n_lipschitz}
		\begin{align}
			\Lip_x(\sfGn) 
            & \leq \Lip_x(\sfG) \frac{1-[\Lip_z(\sfG)]^n}{1-\Lip_z(\sfG)},  \label{eq:g_n_lipschitz:x} \\
			\Lip_z(\sfGn) &\leq [\Lip_z(\sfG)]^n, \label{eq:g_n_lipschitz:z}
		\end{align} 
	\end{subequations} 
respectively. \hfill $\Box$
\end{lemma}
\vspace{.1cm}
\begin{proof}
     In this proof, norm subscripts are omitted whenever the underlying space is clear from the argument. We prove the bounds by induction. The bounds in \eqref{eq:g_n_lipschitz} are trivially satisfied for $n=1$. Then, suppose the bounds hold for $n=k$, $k\in\mathbb{Z}_{>0}$. Then, focusing first on~\eqref{eq:g_n_lipschitz:x}, compute for $n=k+1$:
    \begin{align*}
        & \hspace{-.7cm} \|\sfGn (\solxone,z)  - \sfGn(\solxtwo, z)\|  \\
        = & \|\sfG(\solxone , \sfG^{k}(\solxone,z))  - \sfG(\solxtwo , \sfG^{k}(\solxtwo,z))\|\\
         \leq &\|\sfG(\solxone , \sfG^{k}(\solxone,z))  - \sfG(\solxtwo , \sfG^{k}(\solxone,z))\|\\
        &+ \|\sfG(\solxtwo , \sfG^{k}(\solxone,z))- \sfG(\solxtwo , \sfG^{k}(\solxtwo,z))\|  \\
         \leq &\Lip_x(\sfG) \| \solxone-\solxtwo\| + \Lip_z(\sfG) \| \sfG^{k}(\solxone,z) - \sfG^{k}(\solxtwo,z) \| \\
         \leq &  (\Lip_x(\sfG) + \Lip_z(\sfG)\Lip_x(\sfG^{k}) )\|\solxone-\solxtwo\|\\
        \leq &   \left(\Lip_x(\sfG) + \Lip_x(\sfG)  \Lip_z(\sfG)\sum_{i=0}^{k-1} [\Lip_z(\sfG)]^i \right)\|\solxone-\solxtwo\|\\
        =&  \left( \Lip_x(\sfG)  \sum_{i=0}^{k} [\Lip_z(\sfG)]^i \right)\|\solxone-\solxtwo\|,
    \end{align*}
    where we use the inductive argument that $\Lip_x(\sfG^{k})\leq  \Lip_x(\sfG)  \sum_{i=0}^{k-1} [\Lip_z(\sfG)]^i $.
    
    Consider now~\eqref{eq:g_n_lipschitz:z}; following similar steps as before, we have that:
   \begin{align*}
 \|\sfGn (x,\solzone)  - \sfGn(x, \solztwo)\| 
        = &\|\sfG(x , \sfG^{k}(x,\solzone))  - \sfG(x , \sfG^{k}(x,\solztwo))\|\\
        \leq & \Lip_z(\sfG)\| \sfG^{k}(x,\solzone) -\sfG^{k}(x,\solztwo)  \|\\
        \leq & \Lip_z(\sfG) \Lip_z(\sfG^{k})\|\solzone -\solztwo  \|\\
        \leq &[\Lip_z(\sfG)]^{k+1}\|\solzone -\solztwo  \|,
        \end{align*}
where the last inequality follows from the inductive assumption. By the definition of minimal Lipschitz constant, we immediately get that the bounds in  \eqref{eq:g_n_lipschitz} also hold for $n=k+1$. This completes the proof by induction. 
\end{proof}

 \vspace{.1cm}

\begin{lemma}[$z^*(x)$ is Lipschitz]
\label{lemma: Lip-of-zstar}
    Consider the map ${\sfG}(x,z)$ 
    and let Assumptions~\ref{as:Lipinterconnection-2} and~\ref{as:contrDiscr-2} hold. Then, the map 
    $x \mapsto z^*(x)$ is Lipschitz and its Lipschitz constant is bounded as
\begin{equation}
\label{eq:z_lipschitz}
	\Lip_x(z^*) \leq \frac{\Lip_x(\sfG)}{1-\Lip_z(\sfG)}.
\end{equation}
where we recall that $\Lip_z(\sfG) < 1$. \hfill $\Box$
\end{lemma}

 \vspace{.1cm}

We note that \eqref{eq:z_lipschitz} in
Lemma~\ref{lemma: Lip-of-zstar} can be obtained
from~\eqref{eq:g_n_lipschitz:x} by taking the limit $n\to+\infty$ (see
also, e.g., \cite[Lemma~1.9]{FB:26-CTDS}). 

\vspace{.1cm}

\begin{lemma}[Conditions for  $\rho\left(M\right) < 1$]
    \label{le:2by2_matrix} Define the matrix $$
    M:=\begin{bmatrix}
            a & b \\ c & d
        \end{bmatrix}
        $$
        where $a,b,c,d\geq 0$. Then,
    \begin{equation*}
        \rho\left(M\right) <1 \Leftrightarrow (1-a)(1-d)> bc, ~a+d < 2 \, .
    \end{equation*}
\end{lemma}
\vspace{.2cm}
\begin{proof}
    The eigenvalues of the matrix $M$ are 
\begin{align*}
        \lambda_{1} = \frac{a + d -\sqrt{(a-d)^2 + 4bc }}{2},\\
        \lambda_{2} = \frac{a + d +\sqrt{(a-d)^2 + 4bc }}{2},
\end{align*}
and they both are  real. Since $|\lambda_1|\leq \lambda_2$,
    we conclude that $\rho(M)<1$ if and only if $\lambda_2=\left(a + d +\sqrt{(a-d)^2 + 4bc }\right)/2<1 $. Then, we have the following: 
    \begin{align*}
        & \hspace{-1.4cm} \left(a + d +\sqrt{(a-d)^2 + 4bc }\right)/2<1 \\\Leftrightarrow ~~ & 2-a-d > \sqrt{(a-d)^2 + 4bc} \\
        \Leftrightarrow ~~ & (2-a-d)^2 > (a-d)^2 + 4bc,~2-a-d> 0 \\
        \Leftrightarrow ~~ & 4ad - 4(a+d) + 4 > 4bc,~2-a-d> 0 \\
        \Leftrightarrow ~~ & (1-a)(1-d) > bc,~2-a-d> 0.
    \end{align*}
This concludes the proof. 
\end{proof}

\vspace{.1cm}

\begin{lemma}
\label{lemma:Input-state stability properties}
Consider the system \eqref{eq:sample_data-plant} and let Assumptions~\ref{as: existence-of-weak-pairing} and \ref{as:Lipinterconnection-2} hold. Let $\xi \in \mathbb{R}$ be such that  $\osLip_x(f)\leq\xi$. Then, any two solutions $\solxone(t)$ and $\solxtwo(t)$ of \eqref{eq:sample_data-plant} with input signals $u$, $\bar{u}$: $\mathbb{R}_{\geq 0} \to \mathcal{Z}$, respectively,  satisfy:
    \begin{align}
    \notag
        \|\solxone(t)-\solxtwo(t)\| &\leq  \e^{ \xi t }\|\solxone(0)-\solxtwo(0)\|\\
        &+\Lip_z(f) \int_0^t \e^{\xi(t-s)}\left\|u(s)-\bar{u}(s)\right\|_{\mathcal{Z}} ds.  \tag*{\text{$\Box$}}
    \end{align}
\end{lemma}

\vspace{.1cm} 

Lemma \ref{lemma:Input-state stability properties} is a simple generalization of \cite[Corollary 3.17]{FB:26-CTDS}. The only difference is that the assumption about negativeness of  $\osLip_x(f(x,z))$ is removed. The proof is omitted. 
Lastly, the following lemma was used in the proof of Theorem \ref{prop:Cont_Disc_21_to_22_general_case}.

\vspace{.1cm} 

\begin{lemma}
\label{lemma: Explicit bound for x(t)- x(0,k-1)}
     Consider the interconnected system \eqref{eq:sample_data} and let Assumptions~\ref{as: existence-of-weak-pairing} and~\ref{as:Lipinterconnection-2} be satisfied. 
     Then, for any solution $[x(t)^\top,z(t)^\top]^\top$  of  \eqref{eq:sample_data}, for any $t\in[(k-1)T,kT)$ and any $\xi\geq \osLip_x(f)$ it holds that: 
     \begin{align*}
         &\|x(t)-x((k-1)T)\|\\
       \leq &  h(T,\xi)\left(\Lip_x(f) \left\|x((k-1)T)\right\|+\Lip_z(f)  \left\|z((k-1)T)\right\| \right),
     \end{align*}
     where function $h$ is defined as \eqref{eq:def-of-h}.

\end{lemma}

\begin{proof}
    We define $r(t;k-1):= x(t)- x((k-1)T),~t\in[(k-1)T,kT)$. Then $r((k-1)T;k-1)=0$ and $r(t;k-1)$ is a solution of the system \begin{equation}\label{eq:def_dot_x-x_k-1}  
        \dot r= \dot x=f(r+v ,u), ~\forall~t\in[(k-1)T,kT)
  \end{equation}
  for $t\in[(k-1)T,kT)$, with constant inputs $u=z((k-1)T)$ and $v=x((k-1)T)$. Also, $r_{\text{zero}}(t)\equiv 0$ is a solution to \eqref{eq:def_dot_x-x_k-1} with two zero inputs.
  Applying Lemma~\ref{lemma:Input-state stability properties} to the two solutions $r(t;k-1)$ and $r_{\text{zero}}(t)$, it follows that  for any $t\in [(k-1)T,kT)$,
\begin{align*}
 &  \|r(t;k-1)-r_{\text{zero}}(t)\| \\
  \leq &\Lip_x(f) \! \int_{(k\!-\!1)T}^{t} \e^{\xi(t-s)}\left\|x((k\!-\!1)T)\right\| \! ds\\
 &~+\Lip_z(f) \! \int_{(k\!-\!1)T}^{t}\!\e^{\xi(t-s)}\left\|z((k\!-\!1)T)\right\| \! ds\\
   \leq & h(T,\xi)\left(\Lip_x(f) \left\|x((k\!-\!1)T)\right\|+\Lip_z(f)  \left\|z((k-1)T)\right\| \right),
 \end{align*}
where the second inequality results from the fact that for any $t\in [(k-1)T,kT)$ and $\tau=t-(k-1)T$ it holds that:
\begin{align*}
     \int_{(k\!-\!1)T}^{t} \e^{\xi(t-s)} ds= & \int_{0}^{\tau} \e^{\xi(\tau-s)}ds= h(\tau,\xi)\leq h(T,\xi) \, .
\end{align*}
\end{proof}

\subsection{Proof of Proposition~\ref{prop:contractivity-implies-GES}}

To prove the result of Proposition~\ref{prop:contractivity-implies-GES}, recall first that 
for any $t\geq 0$, we write $t=kT+\tau$, $k\in\mathbb{Z}_{\geq 0}$, $\tau\in [0,T)$. By Lemma \ref{lem:bound-x-by-xkT}, it holds that
\begin{align*}
&\begin{bmatrix}
  \left\| x(\tau + kT)\right\|_\mathcal{X}\\
        \left\| z(\tau + kT)\right\|_\mathcal{Z}
    \end{bmatrix}
    \leq \mathcal{B}(T)
    \begin{bmatrix}
    \left\| x(kT)\right\|_\mathcal{X}\\
        \left\| z(kT)\right\|_\mathcal{Z}
    \end{bmatrix} 
\end{align*}
for any $\tau\in [0,T)$, where $h(T,\xi)$ is defined in \eqref{eq:def-of-h} and $\xi \in \mathbb{R}$ is such that $\xi\geq \osLip_x(f(x,z))$. Recall that $\| \cdot\|$ is a monotonic norm on $\mathbb{R}^2$, and that the composite norm $\|\cdot\|_{\textsf{cmp}}$ on $\mathcal{X}\times\mathcal{Z}$ is defined as $\|y\|_{\textsf{cmp}}:=\| [\| x\|_\mathcal{X},\| z\|_\mathcal{Z} ]^\top\|$.
Then, it follows that
\begin{align*}
&~~~\left\|\begin{bmatrix}
     x(t)\\z(t)
\end{bmatrix}\right\|_{\textsf{cmp}}
 =\left\|\begin{bmatrix}
  \left\| x(kT+\tau)\right\|_\mathcal{X}\\
        \left\| z(kT+\tau)\right\|_\mathcal{Z}
    \end{bmatrix}\right\|\\
&    \leq  \left\|\mathcal{B}(T)\right\|
   \left\|\begin{bmatrix}
  \left\| x(kT)\right\|_\mathcal{X}\\
        \left\| z(kT)\right\|_\mathcal{Z}
    \end{bmatrix}\right\|  \leq  b^k \left\| \mathcal{B}(T)\right\| 
   \left\|\begin{bmatrix}
  \left\| x(0)\right\|_\mathcal{X}\\
        \left\| z(0)\right\|_\mathcal{Z}
    \end{bmatrix}\right\|\\
    &  =  
    \frac{
    \left\| \mathcal{B}(T)
    \right\|} {b^\frac{\tau}{T}} 
    \left(b^\frac{1}{T}\right)^{kT+\tau}
 \left\|\begin{bmatrix}
  \left\| x(0)\right\|_\mathcal{X}\\
        \left\| z(0)\right\|_\mathcal{Z}
    \end{bmatrix}\right\|\\
    &  \leq  r \e^{-c(kT+\tau)}  \left\|\begin{bmatrix}
     x(0)\\z(0)
\end{bmatrix}\right\|_{\textsf{cmp}},
\end{align*}
where   $r:=\left\|\mathcal{B}(T)
    \right\|/b$ and $c:=- \ln  b^\frac{1}{T}=-\frac{1}{T}\ln b>0$. We note that the first inequality follows the monotonicity of the norm $\|\cdot\|$; in the second inequality, we use the fact that the system is contractive in discrete time with composite norm $\left\|\cdot\right\|_{\textsf{cmp}}$ and sampling time $T$; the third inequality holds by $0<b<1$ and $b^{-\frac{\tau}{T}}<b^{-1}$.

\subsection{Proof of Theorem \ref{prop:Cont_Disc_11_to_12_general_case}}

To keep the notation light without losing information, we define $\xi=\osLip_x(f)$, $\tx(t):=\solxone(t)-\solxtwo(t)$ and $\tz(t)=\solzone(t)-\solztwo(t)$. We also drop the subscript from the norm. By Lemma \ref{lem:bound-x-by-xkT} and the continuity of $\tx(t)$, it follows that:
\begin{align*}
   & \hspace{-.3
   cm} \|\tx(kT) \|=  \lim_{\tau\to T_{-}}\|\tx(\tau + (k-1)T) \|\\
   \leq&  \lim_{\tau\to T_{-}} \e^{\xi\tau}\|\tx((k-1)T) \|+ 
 \Lip_z(f)\frac{1-\e^{\xi\tau}}{-\xi}\|\tz((k-1)T) \| \\
   =& \e^{\xi T}\|\tx((k-1)T) \|+ 
 \Lip_z(f)\frac{1-\e^{\xi T}}{-\xi}\|\tz((k-1)T) \|.
\end{align*}

Next, we bound $\|\tz(kT)\|$. By Lemma~\ref{lemma: Lip-of-gn}, it holds that:
\begin{align*}
    &\|\tz(kT) \|\\
    & \leq  \frac{1-\Lip_z(\sfG)^n}{1-\Lip_z(\sfG)}\Lip_x(\sfG)\|\tx(kT) \|+\Lip_z(\sfG)^n\|\tz((k-1)T) \|\\
    & \leq  \frac{1-\Lip_z(\sfG)^n}{1-\Lip_z(\sfG)}\Lip_x(\sfG) \Big(
 \Lip_z(f)\frac{1-\e^{\xi T}}{-\xi}\|\tz((k-1)T) \| \\
    &~~~+  
    \e^{\xi T}\|\tx((k-1)T) \|\Big)+\Lip_z(\sfG)^n\|\tz((k-1)T) \|\\
 & = \frac{1-\Lip_z(\sfG)^n}{1-\Lip_z(\sfG)}\Lip_x(\sfG) \e^{\xi T} \|\tx((k-1)T) \|\\
 &+\frac{1-\Lip_z(\sfG)^n}{1-\Lip_z(\sfG)}\Lip_x(\sfG)  
 \Lip_z(f)\frac{1-\e^{\xi T}}{-\xi}\|\tz((k-1)T) \|\\
 &+\Lip_z(\sfG)^n  \|\tz((k-1)T) \| \, .
\end{align*}
Using the bounds on $\tx(kT)$ and $\tz(kT)$ above, we can write the following:
\begin{equation}\label{eq: small gain_iterative_bound}
    \begin{bmatrix}
    \left\| \tx(kT)\right\|_\mathcal{X}\\
        \left\| \tz(kT)\right\|_\mathcal{Z}
    \end{bmatrix}\leq \mathcal{A}(n,T) \begin{bmatrix}
    \left\| \tx((k-1)T)\right\|_\mathcal{X}\\
        \left\| \tz((k-1)T)\right\|_\mathcal{Z}
    \end{bmatrix},
\end{equation}
where $\mathcal{A}(n,T)$ is a $2\times 2$ matrix defined in Remark~\ref{rem:No-constraints-T-n}. 
We leverage Lemma \ref{le:2by2_matrix} to show that the matrix $\mathcal{A}(n,T)$ is Schur. To this end, we first inspect  $\mathcal{A}_{11}+\mathcal{A}_{22}$. Compute,
 \begin{align*}
\mathcal{A}_{11} +\mathcal{A}_{22} &  =  \e^{\xi T}+\Lip_z(\sfG)^n\\
   & ~~~~ + (1-\Lip_z(\sfG)^n)(1-\e^{\xi T})\frac{\Lip_x(\sfG)  
 \Lip_z(f)}{-\xi (1-\Lip_z(\sfG))}\\
 & <   \e^{\xi T}+  (1-\Lip_z(\sfG)^n)(1-\e^{\xi T})+\Lip_z(\sfG)^n\\
 & = 1+\Lip_z(\sfG)^n \e^{\xi T} \, .
   \end{align*}
We will use this inequality shortly. Next, it holds that 
\begin{align*}
    & \hspace{-1.0cm} (1-\mathcal{A}_{11})(1-\mathcal{A}_{22})-\mathcal{A}_{12}\mathcal{A}_{21}>0\\
     \Longleftrightarrow~~& \frac{1-\mathcal{A}_{11}}{1-\e^{\xi T} }(1-\mathcal{A}_{22})-\frac{\mathcal{A}_{12}\mathcal{A}_{21}}{1-\e^{\xi T}}>0\\
 \Longleftrightarrow~~&  
 1-\Lip_z(\sfG)^n -\frac{\Lip_x(\sfG)  
 \Lip_z(f)}{-\xi (1-\Lip_z(\sfG))}(1-\Lip_z(\sfG)^n) \e^{\xi T} \\
&-\frac{\Lip_x(\sfG)  
 \Lip_z(f)}{-\xi(1-\Lip_z(\sfG))}(1-\Lip_z(\sfG)^n)(1-\e^{\xi T}) >0\\
 \Longleftrightarrow~~&   
 1-\Lip_z(\sfG)^n-(1-\Lip_z(\sfG)^n)\frac{\Lip_x(\sfG)  
 \Lip_z(f)}{-\xi(1-\Lip_z(\sfG))}  >0\\
  \Longleftrightarrow~~&   
 \left(1-\frac{\Lip_x(\sfG)  
 \Lip_z(f)}{-\xi(1-\Lip_z(\sfG))} \right) (1-\Lip_z(\sfG)^n) >0\\
  \Longleftrightarrow~~&   
 -\xi(1-\Lip_z(\sfG))>\Lip_x(\sfG)  
 \Lip_z(f) \, .
 \end{align*}
Given the last inequality,  one has that $\xi < 0$; therefore, $1+\Lip_z(\sfG)^n \e^{\xi T} < 2$. It then follows from Lemma \ref{le:2by2_matrix} that $\mathcal{A}(n,T)$ is Schur for any $n\in\mathbb{Z}_{>0}$ and any $T>0$.

To show the contractivity in discrete time, we note that  all the entries of  $\mathcal{A}(n,T)$ are positive, by the fact that $\xi<0$, $\Lip_z(\sfG)<1$, $\Lip_z(f)\neq 0,~\Lip_x(\sfG)\neq 0$. Hence  $\mathcal{A}(n,T)$ is irreducible. By \cite[Lemma 2.8, 2.22]{FB:26-CTDS},
  it follows that there exists a vector $\eta=[\eta_1,\eta_2]^\top\in\mathbb{R}^2$ such that the vector norm  $\|\cdot \|_{2,[\eta]}$ is monotonic; and (ii) the induced matrix norm of $\|\cdot \|_{2,[\eta]}$ satisfies that $\| \mathcal{A}(n,T) \|_{2,[\eta]}=\rho\left( \mathcal{A}(n,T)\right)$, where the vector norm  $\|\cdot \|_{2,[\eta]}$ is defined as   $\|[v_1,v_2]^\top \|_{2,[\eta]}:=\sqrt{\eta_1 v_1^2+\eta_2 v_2^2}$. 
  
  Taking $\|\cdot \|_{2,[\eta]}$ on both sides of \eqref{eq: small gain_iterative_bound}, then by monotonicity of the norm $\|\cdot \|_{2,[\eta]}$, it follows that
  \begin{align*}
    &\left\|   \begin{bmatrix}
    \left\| \tx(kT)\right\|_\mathcal{X}\\
        \left\| \tz(kT)\right\|_\mathcal{Z}
    \end{bmatrix}\right\|_{2,[\eta]}\\
    \leq& \! \left\|\mathcal{A}(n,T) \right\|_{2,[\eta]}\! \left\|  \begin{bmatrix}
    \left\| \tx((k\!-\!1)T)\right\|_\mathcal{X}\\
        \left\| \tz((k\!-\!1)T)\right\|_\mathcal{Z}
    \end{bmatrix}\right\|_{2,[\eta]}\\
    \leq & \! \left\|\mathcal{A}(n,T) \right\|_{2,[\eta]}^k  \left\|  \begin{bmatrix}
    \left\| \tx(0)\right\|_\mathcal{X}\\
        \left\| \tz(0)\right\|_\mathcal{Z}
    \end{bmatrix}\right\|_{2,[\eta]}\\
    =& \! \left[\rho\left( \mathcal{A}(n,T)\right)\right]^k  \left\|  \begin{bmatrix}
    \left\| \tx(0)\right\|_\mathcal{X}\\
        \left\| \tz(0)\right\|_\mathcal{Z}
    \end{bmatrix}\right\|_{2,[\eta]}.
\end{align*}
Hence, for any $n\in\mathbb{Z}_{>0}$ and any sampling time $T>0$, the system is contractive in discrete time at rate $ \rho\left( \mathcal{A}(n,T)\right)$, with composite norm $\|[x^\top,z^\top]^\top\|_{\textsf{cmp}}:=\| [\| x\|_\mathcal{X},\| z\|_\mathcal{Z} ]^\top\|_{2,[\eta]}$.

To show GES, we recall that the norm $\| \cdot\|_{2,[\eta]}$ is monotonic. Then, by Proposition \ref{prop:contractivity-implies-GES}, it holds that for $n\in\mathbb{Z}_{>0}$, any $T>0$, and any solution $[x(t)^\top,z(t)^\top]$ to the system \eqref{eq:sample_data},
\begin{equation*}
     \left\|   \begin{bmatrix}
    \left\| \tx(t)\right\|_\mathcal{X}\\
        \left\| \tz(t)\right\|_\mathcal{Z}
    \end{bmatrix}\right\|_{2,[\eta]}\leq r\e^{-ct}  \left\|  \begin{bmatrix}
    \left\| \tx(0)\right\|_\mathcal{X}\\
        \left\| \tz(0)\right\|_\mathcal{Z}
    \end{bmatrix}\right\|_{2,[\eta]},~\forall~t\geq 0,
\end{equation*}
where $r:=\left\| \mathcal{B}(T)
    \right\|_{2,[\eta]}/\rho\left( \mathcal{A}(n,T)\right)$ and $c:=-\frac
{1}{T}\ln \rho\left( \mathcal{A}(n,T)\right)$. This completes the proof.

\subsection{Proof of Theorem \ref{prop:Cont_Disc_11_to_21_general_case}}

We need to show that the assumptions in Theorem~\ref{prop:Cont_Disc_11_to_21_general_case} imply that $\osLip(f(x,z^*(x))) < 0$ for any $x$. To this end, recall the definition of one-sided Lipschitz constant: 
\begin{align*}
\osLip_x(f(x,z^*(x)))
    \!=\! \sup_{\solxone \neq \solxtwo} \! \frac{\llbracket f(\solxone,z^*(\solxone))\!-\!f(\solxtwo,z^*(\solxtwo)) ; \solxone\!-\!\solxtwo \rrbracket_{\mathcal{X}}}{\|\solxone-\solxtwo\|^2_{\mathcal{X}}}
\end{align*}
where we use a  weak pairing  $\llbracket \cdot ; \cdot \rrbracket_\mathcal{X} $ compatible with the norm $\| \cdot\|_\mathcal{X}$; we also recall that a compatible weak pairing is such that $\llbracket x ; x \rrbracket_\mathcal{X} = \|x\|_\mathcal{X}^2$~\cite[Ch.~2]{FB:26-CTDS}. Then, we have the following: 
\begin{align*}
    &\osLip(f(x,z^*(x)))\\
    \leq & \sup _{\solxone \neq \solxtwo} \frac{\llbracket f(\solxone,z^*(\solxone))-f(\solxtwo,z^*(\solxone)) ; \solxone-\solxtwo \rrbracket_{\mathcal{X}}}{\|\solxone-\solxtwo\|^2_{\mathcal{X}}} \\
    &+\sup _{\solxone \neq \solxtwo} \frac{\llbracket f(\solxtwo,z^*(\solxone))-f(\solxtwo,z^*(\solxtwo)) ; \solxone-\solxtwo \rrbracket_{\mathcal{X}}}{\|\solxone-\solxtwo\|^2_{\mathcal{X}}} \\
    \leq&~\osLip_x(f)\\
    &+\sup _{\solxone \neq \solxtwo} \frac{\| f(\solxtwo,z^*(\solxone))-f(\solxtwo,z^*(\solxtwo))\|_{\mathcal{X}} \| \solxone-\solxtwo\|_{\mathcal{X}}}{\|\solxone-\solxtwo\|^2_{\mathcal{X}}}    \\
    \leq &~ \osLip_x(f)+\sup _{\solxone \neq \solxtwo} \frac{ \Lip_z(f)  \|z^*(\solxone)-z^*(\solxtwo)  \|_{\mathcal{Z}}}{\|\solxone-\solxtwo\|_{\mathcal{X}}}    \\
    \leq & ~\osLip_x(f)+\frac{\Lip_z(f)\Lip_x(\sfG)}{1-\Lip_z(\sfG)}\\
    =&\frac{\osLip_x(f)(1-\Lip_z(\sfG))+\Lip_z(f)\Lip_x(\sfG)}{1-\Lip_z(\sfG)} < 0 
\end{align*}
where we have used the sub-additivity of the weak pairing, the Cauchy-Schwarz inequality for the weak pairing~\cite[Def.~2.27]{FB:26-CTDS}, and where the last term is negative by~\eqref{eq:small-gain-condition} and the fact that $\Lip_z(\sfG)<1$. This concludes the proof.

\subsection{Proof of Theorem \ref{prop:Cont_Disc_21_to_22_general_case}}

In this proof, define the error variable $e(t) := z(t) - z^*(x(t))$ for brevity. Recall that,  without loss of generality, $f(\mathbf{0}_{n_x},\mathbf{0}_{n_z})=\mathbf{0}_{n_x}$, $\sfG(\mathbf{0}_{n_x},\mathbf{0}_{n_z})=\mathbf{0}_{n_z}$, and $z^*(\mathbf{0}_{n_x}) = \mathbf{0}_{n_z}$. To simplify the notation, we omit the sub-indices for all the norms and log-norms when the context is clear. 

Several steps of the proof will rely on the following bounds to switch back-and-forth between $z$ and $e$: 
\begin{align}
\notag
        \|z(kT) \|\leq&  \|z(kT)-z^*(x(kT)) +z^*(x(kT))-z^*(\mathbf{0}_{n_x})\|\\
        \label{eq: nonlinear bound z by e}
        \leq &\|e(kT)\| +\frac{\Lip_x(\sfG)}{1-\Lip_z(\sfG)}\|x(kT)\| \\
        \notag
        \|e(kT) \|=&  \|z(kT)-z^*(x(kT))+z^*(\mathbf{0}_{n_x}) \|\\
        \label{eq: nonlinear bound e by z}
        \leq&\|z(kT)\| +\frac{\Lip_x(\sfG)}{1-\Lip_z(\sfG)}\|x(kT)\| \ 
\end{align}
Next, let $y(t) = [x(t)^\top,z(t)^\top]^\top$ be an arbitrary solution of interconnected system \eqref{eq:sample_data}. 
During time interval $[(k\!-\!1)T,kT)$, $x(t)$ is also a solution to the following system:
\begin{equation}\label{eq: reduced model with input}
    \dot x(t)=f(x(t),z^*(x(t)))+u(t),  ~~ t \geq 0 \, ,
\end{equation}
where $t \mapsto u(t)$ is defined as $u(t):=f(x(t),z((k\!-\!1)T))-f(x(t),z^*(  x((k\!-\!1)T) ) )+f(x(t),z^*(  x((k\!-\!1)T) ) )-f(x(t),z^*(x(t)))$. In addition, let $x_{\text{zero}}(t)\equiv 0$ be a solution to system \eqref{eq: reduced model with input} with input $u\equiv 0$. 
Next, we apply Lemma \ref{lemma:Input-state stability properties} to bound $\|x(t)-x_{\text{zero}}(t) \|$. We verify that the system \eqref{eq: reduced model with input} satisfies the assumptions in Lemma \ref{lemma:Input-state stability properties}: $\osLip_x(f(x,z^*(x))+u)=\osLip(f(x,z^*(x)))\leq -\zeta$ for any fixed $u$ and the dynamics is Lipschitz w.r.t $u$ with Lipschitz constant $1$. By Lemma \ref{lemma:Input-state stability properties}, it follows that for  $t=(k-1)T+\tau$, $\tau\in[0,T)$,
\begin{equation}   \label{eq: SP_ bound for x} 
\begin{aligned}
    \|x(t) \| & \leq  \e^{ -\zeta \tau }\|x((k\!-\!1)T)\| +\int_{(k-1)T}^t \e^{-\zeta(t-s)}\cdot \overline{u} 
  ~ds\\
 & \leq \e^{ -\zeta \tau }\|x((k-1)T)\|+\frac{1-\e^{-\zeta T} }{\zeta}\overline{u},
\end{aligned}
\end{equation}
where $\overline{u} := \sup_{t\in[(k-1)T,kT)} \|u(t)\|$ and where we used the fact that $\int_{(k-1)T}^t \e^{-\zeta(t-s)} ds\leq\int_{0}^T \e^{-\zeta(T-s)} ds=\frac{1-\e^{-\zeta T} }{\zeta}$. Then it suffices to bound $\overline{u}$; to this end, compute, 
\begin{align*}
    \overline{u}  \leq & \Lip_z(f)\|z((k-1)T)-z^*(  x((k-1)T) )   \|\\
    &+ \sup_{t\in[(k-1)T,kT)}\Lip_z(f)\|z^*(  x((k-1)T) )-z^*(x(t))   \|\\
    \leq &  \Lip_z(f) \|e((k-1)T)   \|\\
    &+ \frac{\Lip_z(f)\Lip_x(\sfG)}{1-\Lip_z(\sfG)}\sup_{t\in[(k-1)T,kT)}\|  x((k-1)T) -x(t) \|\\
    \leq & \Lip_z(f) \|e((k-1)T)   \|\\
    &+ \frac{\Lip_z(f)\Lip_x(\sfG)}{1-\Lip_z(\sfG)} h(T,\xi)\Lip_x(f)\| x((k-1)T) \|\\
    &+\frac{\Lip_z(f)\Lip_x(\sfG)}{1-\Lip_z(\sfG)} h(T,\xi)\Lip_z(f)\| z((k-1)T) \|\\
    \leq & \Lip_z(f) \|e((k-1)T)   \|\\
    &+ \frac{\Lip_z(f)\Lip_x(\sfG)}{1-\Lip_z(\sfG)} \Lip_x(f) h(T,\xi)\| x((k-1)T) \|\\
    &+\left[\frac{\Lip_z(f)\Lip_x(\sfG)}{1-\Lip_z(\sfG)} \right]^2 h(T,\xi)\| x((k-1)T) \|\\
    &+\frac{\Lip_z(f)\Lip_x(\sfG)}{1-\Lip_z(\sfG)} \Lip_z(f)h(T,\xi)\| e((k-1)T) \|\\
    = & h(T,\xi)   C_{1}\| x((k-1)T)\|\\
    &+\left( \Lip_z(f)+h(T,\xi)C_{12} \right) \|e((k-1)T)   \|\\
\end{align*}
where we use Lemma \ref{lemma: Explicit bound for x(t)- x(0,k-1)} in the third inequality and \eqref{eq: nonlinear bound z by e} in the fourth inequality. Combining \eqref{eq: SP_ bound for x} and the upper bound for $\overline{u} $, we can write 
\begin{equation}\label{eq: SP_bound for x, final}    
\begin{aligned}
 &\|x(kT) \|=\lim_{\tau\to T_{-}} \|x(\tau+ (k-1)T) \| \\
 & ~~ \leq \left(\e^{-\zeta T}+    \frac{1-\e^{-\zeta T}}{\zeta}h(T,\xi) C_{1} \right) \left \|x((k-1)T)\right\|\\
 & ~~~~ +  \frac{1-\e^{-\zeta T}}{\zeta}(\Lip_z(f)+ h(T,\xi)C_{12})\|e((k-1)T)\|
\end{aligned}
\end{equation}
with $C_{1}$ and $C_{12}$ defined in Theorem~\ref{prop:Cont_Disc_21_to_22_general_case} and Remark~\ref{rem:transientbound}.

Next, we bound $\|e(kT)\|$. Since $z^*(x)$ is a fixed point of $\sfGn(x, \cdot)$ for any $x$ and any $n\in\mathbb{Z}_{>0}$, it follows that
\begin{equation}\label{eq: SP bound for e}
\begin{aligned}
     &\hspace{-.5cm}\|e(kT)\| =\|z(kT) -z^*( x(kT) )\|\\
     =& \|\sfGn (x(kT),z((k-1)T))-\sfGn( x(kT),z^*(x(kT)) )\|  \\
     \leq & [\Lip_z(\sfG)]^n  \|z((k-1)T)-z^*(x(kT))  \|    \\
      \leq & [\Lip_z(\sfG)]^n  \|z((k-1)T)-z^*(x ((k-1)T) )\|\\
      &+ [\Lip_z(\sfG)]^n\| z^*(x((k-1)T))-z^*(x(kT))  \|              \\
      \leq &  [\Lip_z(\sfG)]^n \| e((k-1)T)\|\\
      &+  [\Lip_z(\sfG)]^n \frac{\Lip_x(\sfG)}{1-\Lip_z(\sfG)}\|x(kT)- x((k-1)T) \|, 
\end{aligned}
\end{equation}
where we use Lemma \ref{lemma: Lip-of-gn} and Lemma \ref{lemma: Lip-of-zstar}. By Lemma~\ref{lemma: Explicit bound for x(t)- x(0,k-1)}, we have 
\begin{align*}
    &\hspace{-.3cm} \|x(kT)- x((k-1)T) \| = \lim_{t\to (kT)_{-}}  \|x(t)- x((k-1)T) \|\\
    \leq&h(T,\xi)( \Lip_x(f)\|  x((k-1)T) \|+\Lip_z(f) \|z((k-1)T)\|)\\
    \leq &\left( \Lip_x(f)+\frac{\Lip_x(\sfG)\Lip_z(f)}{1-\Lip_z(\sfG)} \right)  h(T,\xi) \|  x((k-1)T) \|\\
    &+\Lip_z(f)  h(T,\xi) \|e((k-1)T)\| \, .
\end{align*}
Then, combining \eqref{eq: SP bound for e} and the upper bound for $\|x(kT)- x((k-1)T) \|$, we can write
\begin{equation}\label{eq: SP bound for e, final}
\begin{aligned}
    \|e(kT) \|\leq &(h(T,\xi)C_{2} +[\Lip_z(\sfG)]^n) \|  e((k-1)T) \|\\
    &+h(T,\xi) C_{21} \|  x((k-1)T) \|
\end{aligned}
\end{equation}
where the constant $C_{21}$ and $C_{2}$ are defined as in Theorem~\ref{prop:Cont_Disc_21_to_22_general_case} and Remark~\ref{rem:transientbound}.

By \eqref{eq: SP_bound for x, final} and \eqref{eq: SP bound for e, final}, it follows that 
\begin{equation}
\begin{bmatrix}\label{eq: iterative_bound_for x and e}
    \left\| x(kT)\right\|\\
        \left\| e(kT)\right\|
    \end{bmatrix}\leq \overline{\mathcal{A}}(n,T) \begin{bmatrix}
    \left\| x((k-1)T)\right\|\\
        \left\| e((k-1)T)\right\|
    \end{bmatrix}
\end{equation}
    with $ \overline{\mathcal{A}}(n,T)$ defined in Remark~\ref{rem:transientbound}. Next, we use Lemma \ref{le:2by2_matrix} to derive a sufficient condition for $\overline{\mathcal{A}}(n,T)$ to be Schur stable. Compute,
\begin{align*}    
& \hspace{-.4cm} \overline{\mathcal{A}}_{11} +\overline{\mathcal{A}}_{22}\\
=& \e^{-\zeta T}+    \frac{1-\e^{-\zeta T}}{\zeta}h(T,\xi) C_{1}+  h(T,\xi)C_{2} +[\Lip_z(\sfG)]^n\\
< &1+    \frac{1}{\zeta}h(T,\xi) C_{1}+  h(T,\xi)C_{2} +[\Lip_z(\sfG)]^n
\end{align*}
where we use $\e^{-\zeta T}<1$ and $1-\e^{-\zeta T}< 1$ in the first inequality; hence, a sufficient condition for $\overline{\mathcal{A}}_{11} +\overline{\mathcal{A}}_{22}<2$ is
\begin{equation}\label{eq: condition for T(n)}
    h(T,\xi)(C_{2}+C_{1}/\zeta)+[\Lip_z(\sfG)]^n-1<0.
\end{equation}

In addition, we note that
    \begin{align*}
    &(1-\overline{\mathcal{A}}_{11})(1-\overline{\mathcal{A}}_{22})-\overline{\mathcal{A}}_{12}\overline{\mathcal{A}}_{21}>0\\
     \Longleftrightarrow~~& \frac{1-\overline{\mathcal{A}}_{11}}{1-\e^{-\zeta T} }(1-\overline{\mathcal{A}}_{22})-\frac{\overline{\mathcal{A}}_{12}\overline{\mathcal{A}}_{21}}{1-\e^{-\zeta T}}>0\\
  \Longleftrightarrow~~&\left(1- \frac{1}{\zeta} h(T,\xi)C_{1} \right) 
  \left(1-   h(T,\xi)C_{2}-[\Lip_z(\sfG)]^n \right)\\
  ~~&- \frac{1}{\zeta}(\Lip_z(f)+  h(T,\xi)C_{12})h(T,\xi)C_{21} >0, \\
\Longleftrightarrow~~&\left(1- \frac{1}{\zeta} h(T,\xi)C_{1} \right) 
  \left(1-[\Lip_z(\sfG)]^n \right)\\
  ~~&- h(T,\xi)C_{2} - \frac{1}{\zeta} \Lip_z(f)h(T,\xi)C_{21}\\
~~&  +   \frac{1}{\zeta} [h(T,\xi)]^2C_{1}C_{2}- \frac{1}{\zeta} [h(T,\xi)]^2C_{12}C_{21}>0\\ 
\Longleftrightarrow~~&\left(1- \frac{1}{\zeta} h(T,\xi)C_{1} \right) 
  \left([\Lip_z(\sfG)]^n-1 \right)\\
 ~~&+ h(T,\xi)C_{2} + \frac{1}{\zeta} \Lip_z(f)h(T,\xi)C_{21}<0\\
 \Longleftrightarrow~~&h(T,\xi)(C_{2}+C_{1}/\zeta)+[\Lip_z(\sfG)]^n-1\\
 ~~&- \frac{1}{\zeta} h(T,\xi)(C_{1} 
  [\Lip_z(\sfG)]^n-\Lip_z(f)C_{21})<0\\
  \Longleftrightarrow~~&h(T,\xi)(C_{2}+C_{1}/\zeta)+[\Lip_z(\sfG)]^n-1<0,
\end{align*}
where we used the fact that $C_{1}C_{2}=C_{12}C_{21}$ and $C_{1}[\Lip_z(\sfG)]^n=\Lip_z(f)C_{21}$ (the latter can be verified from the definitions of $C_{1}$, $C_{12}$, $C_{21}$ and $C_{2}$). 
Hence, \eqref{eq: condition for T(n)} is a sufficient and necessary condition for $(1-\overline{\mathcal{A}}_{11})(1-\overline{\mathcal{A}}_{22})-\overline{\mathcal{A}}_{12}\overline{\mathcal{A}}_{21}>0$.
By Lemma \ref{le:2by2_matrix}, we conclude that $\overline{\mathcal{A}}(n,T)$ is Schur if 
\begin{equation}\label{eq: condition II for T(n)}
    h(T,\xi)< \frac{1-[\Lip_z(\sfG)]^n}{C_{2}+C_{1}/\zeta}.
\end{equation}
Based on  \eqref{eq:def-of-h} and the monotonicity of $h$, we get that \eqref{eq: condition II for T(n)} holds for any $T<T(n)$, where $T(n)$ is defined as \eqref{eq:Tn}.

To conclude the proof for GES, notice that all the entries of  $\overline{\mathcal{A}}(n,T)$ are positive and that $\overline{\mathcal{A}}(n,T)$ is irreducible. By \cite[Lemma 2.8, 2.22]{FB:26-CTDS},
  it follows that there exists a vector $\bar{\eta}=[\bar{\eta}_1,\bar{\eta}_2]^\top\in\mathbb{R}^2$ such that the vector norm  $\|\cdot \|_{2,[\bar{\eta}]}$ is monotonic; and (ii) the induced matrix norm of $\|\cdot \|_{2,[\bar{\eta}]}$ satisfies that $\| \overline{\mathcal{A}}(n,T) \|_{2,[\bar{\eta}]}=\rho\left( \overline{\mathcal{A}}(n,T)\right)$, where the vector norm  $\|\cdot \|_{2,[\bar{\eta}]}$ is define as   $\|[v_1,v_2]^\top \|_{2,[\bar{\eta}]}:=\sqrt{\bar{\eta}_1 v_1^2+\bar{\eta}_2 v_2^2}$. 

Define the matrix $L := \begin{bmatrix}
        1 & 0 \\
        \frac{\Lip_x(\sfG)}{1-\Lip_z(\sfG)} & 1
    \end{bmatrix} 
$ for brevity. Applying Lemma \ref{lem:bound-x-by-xkT}, \eqref{eq: nonlinear bound z by e}, \eqref{eq: iterative_bound_for x and e} and \eqref{eq: nonlinear bound e by z} sequentially, we get 
  \begin{align*}
    \begin{bmatrix}
    \left\| x(kT+\tau)\right\|\\
        \left\| z(kT+\tau)\right\|
    \end{bmatrix}
    & \leq   \mathcal{B}(T)   \begin{bmatrix}
    \left\| x(kT)\right\|\\
        \left\| z(kT)\right\|
    \end{bmatrix}\\
      &  \leq  \mathcal{B}(T) L   \begin{bmatrix}
    \left\| x(kT)\right\|\\
        \left\| e(kT)\right\|
    \end{bmatrix}\\
    & \leq   \mathcal{B}(T) L\left[\overline{ \mathcal{A}}(n,T)\right]^k\begin{bmatrix}
    \left\| x(0)\right\|\\
        \left\| e(0)\right\|
    \end{bmatrix}\\
       & \leq   \mathcal{B}(T) L \left[\overline{ \mathcal{A}}(n,T)\right]^k L\begin{bmatrix}
    \left\| x(0)\right\|\\
        \left\| z(0)\right\|
    \end{bmatrix} \, .
  \end{align*} 
  By monotonicity of the norm $\|\cdot \|_{2,[\bar{\eta}]}$,  it follows that
  \begin{align*}
    & \hspace{-.8cm} \left\|   \begin{bmatrix}
    \left\| x(kT+\tau)\right\|\\
        \left\| z(kT+\tau)\right\|
    \end{bmatrix}\right\|_{2,[\bar{\eta}]} \leq \left\| \mathcal{B}(T) L \left[\overline{ \mathcal{A}}(n,T)\right]^k  L\begin{bmatrix}
    \left\| x(0)\right\|\\
        \left\| z(0)\right\|
    \end{bmatrix}\right\|_{2,[\bar{\eta}]} \\
\leq &  \left\| \mathcal{B}(T)\right\|_{2,[\bar{\eta}]} \left\| L\right\|_{2,[\bar{\eta}]}^2  \left\|\overline{ \mathcal{A}}(n,T)\right\|_{2,[\bar{\eta}]}^k \left\|\begin{bmatrix}
    \left\| x(0)\right\|\\
        \left\| z(0)\right\|
    \end{bmatrix}\right\|_{2,[\bar{\eta}]}\\
    = & Q \cdot b^k \cdot \left\|\begin{bmatrix}
    \left\| x(0)\right\|\\
        \left\| z(0)\right\|
    \end{bmatrix}\right\|_{2,[\bar{\eta}]}\\
     = & \frac{Q}{b^{\frac{\tau}{T}}}  \left(b^{\frac{1}{T}}\right)^{kT+\tau}  \left\|\begin{bmatrix}
    \left\| x(0)\right\|\\
        \left\| z(0)\right\|
    \end{bmatrix}\right\|_{2,[\bar{\eta}]}\\
\leq &  \varrho \e^{- \alpha (kT+\tau)} \left\|\begin{bmatrix}
    \left\| x(0)\right\|\\
        \left\| z(0)\right\|
    \end{bmatrix}\right\|_{2,[\bar{\eta}]},
  \end{align*}
where $Q:= \left\| \mathcal{B}(T)\right\|_{2,[\bar{\eta}]} \left\| L \right\|_{2,[\bar{\eta}]}^2$, $b:=\rho\left( \overline{\mathcal{A}}(n,T)\right)$, $\varrho:=\frac{Q}{b}$ and $\alpha:=-\frac{1}{T}\log b$. This concludes the proof.

\section{Conclusions}
\label{sec:conclusions}

This paper examined sampled-data systems, formed by the interconnection of a CT dynamical system and a  DT system that updates its output at each interval $T\!>\!0$ based on an  $n$-fold composition of a given map. We formalized the notion of reduced model associated with such a sampled-data system, which was defined as the limiting system for $T\!\rightarrow \!0^+$ and $n\!\rightarrow \!+\infty$. Our main results (Theorem~\ref{prop:Cont_Disc_21_to_22_general_case}) showed that, if the reduced model is contractive, then for each $n$ there exists a $T(n)$ such that the sampled-data system is  exponentially stable for each $T < T(n)$; the theorem provides an explicit (possibly conservative) expression for $T(n)$ as a function of relevant Lipschitz constants. 
This result requires a weaker condition than the exponential‐stability conditions derived via the small‐gain arguments (Theorem~\ref{prop:Cont_Disc_11_to_12_general_case}), which constitute an additional contribution of this work.  
Our findings were then applied to MPC frameworks, offering new insights into the stability of suboptimal MPC. Importantly, our results provide conditions for stability even for a single-iteration suboptimal MPC algorithm; to the best of our knowledge, this is the first time that such an MPC result is established. Future research will further investigate local versions of our results (expanding on Corollary~\ref{corollary:local_Cont_Disc_21_to_22_general_case}), exploring conditions for local contractivity and exponential stability,  with an emphasis on MPC formulations with input and state constraints. Another avenue for future research involves 
the design of observers and output-feedback controllers for interconnected CT–DT 
systems, extending ideas from related work such as \cite{slotine2001modularity}.

\section*{References}
\def\refname{\vadjust{\vspace*{-2.5em}}} 
\vspace{-0.2cm}
\bibliographystyle{IEEEtran}
\bibliography{alias,biblio}

@STRING{automatica = "Automatica"}

@String{SIAM = "SIAM"}

@STRING{springer = "Springer"}

@String{nonlinearity = "Nonlinearity"}

@article{slotine2001modularity,
  title={Modularity, evolution, and the binding problem: a view from stability theory},
  author={Slotine, J-JE and Lohmiller, Winfried},
  journal={Neural networks},
  volume={14},
  number={2},
  pages={137--145},
  year={2001},
  publisher={Elsevier}
}

@article{nguyen2020contraction,
  title={Contraction analysis of nonlinear {DAE} systems},
  author={Nguyen, Hung D and Vu, Thanh Long and Slotine, Jean-Jacques and Turitsyn, Konstantin},
  journal={IEEE Transactions on Automatic Control},
  volume={66},
  number={1},
  pages={429--436},
  year={2020},
  publisher={IEEE}
}

@Book{FB:26-CTDS,
  author =	 {F. Bullo},
  title =	 {Contraction Theory for Dynamical Systems},
  year =	 2026,
  edition =	 {{1.3}},
  publisher =	 {Kindle Direct Publishing},
  ISBN =	 {979-8836646806},
  ISBN-hardocver = {979-8333395283},
  url =		 {https://fbullo.github.io/ctds},
  oldurl =	 {http://motion.me.ucsb.edu/book-ctds},
  funding =	 {FA9550-22-1-0059}
}

@book{khalil2002nonlinear,
  title={Nonlinear Systems},
  author={Khalil, Hassan K},
  publisher={Prentice Hall},
  year={2002}
}

@incollection{heemels2010stability,
  title={Stability and stabilization of networked control systems},
  author={Heemels, WPMH and Van De Wouw, N},
  booktitle={Networked Control Systems},
  pages={203--253},
  year={2010},
  publisher={Springer}
}

@article{caccioppoli1930teorema,
  title =	 {Un teorema generale sull’esistenza di elementi uniti in
                  una transformazione funzionale},
  author =	 {Caccioppoli, Renato},
  journal =	 {Rendiconti dell'Accademia Nazionale dei Lincei},
  volume =	 11,
  pages =	 {794--799},
  year =	 1930
}

@article{banach1922operations,
  title={Sur les op{\'e}rations dans les ensembles abstraits et leur application aux {\'e}quations int{\'e}grales},
  author={Banach, Stefan},
  journal={Fundamenta Mathematicae},
  volume={3},
  number={1},
  pages={133--181},
  year={1922}
}

@article{russo2012contraction,
  title={A contraction approach to the hierarchical analysis and design of networked systems},
  author={Russo, Giovanni and Di Bernardo, Mario and Sontag, Eduardo D},
  journal={IEEE Trans. on Automatic Control},
  volume={58},
  number={5},
  pages={1328--1331},
  year={2012},
  publisher={IEEE}
}

@article{picallo2022sensitivity,
  title={Sensitivity conditioning: Beyond singular perturbation for control design on multiple time scales},
  author={Picallo, Miguel and Bolognani, Saverio and D{\"o}rfler, Florian},
  journal={IEEE Trans. on Automatic Control},
  volume={68},
  number={4},
  pages={2309--2324},
  year={2022},
  publisher={IEEE}
}

@article{hall2024stability,
  title={Stability Certificates for Receding Horizon Games},
  author={Hall, Sophie and Liao-McPherson, Dominic and Belgioioso, Giuseppe and D{\"o}rfler, Florian},
  journal={arXiv preprint arXiv:2404.12165},
  year={2024}
}

@article{kokotovic1968singular,
  title={Singular perturbation method for reducing the model order in optimal control design},
  author={Kokotovic, Petar and Sannuti, Peddapullaiah},
  journal= {IEEE Trans. on Automatic Control},
  volume={13},
  number={4},
  pages={377--384},
  year={1968},
  publisher={IEEE}
}

@article{kokotovic1976singular,
  title={Singular perturbations and order reduction in control theory—an overview},
  author={Kokotovic, Petar V and O'Malley Jr, Robert E and Sannuti, Peddapullaiah},
  journal={Automatica},
  volume={12},
  number={2},
  pages={123--132},
  year={1976},
  publisher={Elsevier}
}

@article{del2012contraction,
  title={A contraction theory approach to singularly perturbed systems},
  author={Del Vecchio, Domitilla and Slotine, Jean-Jacques E},
  journal={IEEE Trans. on Automatic Control},
  volume={58},
  number={3},
  pages={752--757},
  year={2012},
  publisher={IEEE}
}

@article{cothren2023singular,
  title={Online Feedback Optimization and Singular Perturbation via Contraction Theory},
  author={Cothren, Liliaokeawawa and Bullo, Francesco and Dall'Anese, Emiliano},
  journal={arXiv preprint arXiv:2310.07966},
  year={2023}
}

@article{teel2003unified,
  title={A unified framework for input-to-state stability in systems with two time scales},
  author={Teel, Andrew R and Moreau, Luc and Nesic, Dragan},
  journal={IEEE Trans. on Automatic Control},
  volume={48},
  number={9},
  pages={1526--1544},
  year={2003},
  publisher={IEEE}
}

@book{freeman2008robust,
  title={Robust Nonlinear Control Design: State-Space and {Lyapunov} Techniques},
  author={Freeman, Randy and Kokotovic, Petar V},
  year={2008},
  publisher={Springer},
}

@inproceedings{garg2019control,
  title={Control-{Lyapunov} and control-barrier functions based quadratic program for spatio-temporal specifications},
  author={Garg, Kunal and Panagou, Dimitra},
  booktitle={IEEE Conference on Decision and Control},
  pages={1422--1429},
  year={2019}
}

@book{rawlings2017model,
  title={Model Predictive Control: Theory, Computation, and Design},
  author={Rawlings, J.B. and Mayne, D.Q. and Diehl, M.},
  year={2017},
  publisher={Nob Hill Publishing}
}

@article{diehl2005real,
  title={A real-time iteration scheme for nonlinear optimization in optimal feedback control},
  author={Diehl, Moritz and Bock, Hans Georg and Schl{\"o}der, Johannes P},
  journal={SIAM Journal on Control and Optimization},
  volume={43},
  number={5},
  pages={1714--1736},
  year={2005},
  publisher={SIAM}
}

@article{morari1988model,
  title={Model predictive control: Theory and practice},
  author={Morari, Manfred and Garcia, Carlos E and Prett, David M},
  journal={IFAC Proceedings Volumes},
  volume={21},
  number={4},
  pages={1--12},
  year={1988},
  publisher={Elsevier}
}

@inproceedings{karapetyan2023finite,
  title={On the finite-time behavior of suboptimal linear model predictive control},
  author={Karapetyan, Aren and Balta, Efe C and Iannelli, Andrea and Lygeros, John},
  booktitle={IEEE Conference on Decision and Control},
  pages={5053--5058},
  year={2023}
}

@article{zeilinger2011real,
  title={Real-time suboptimal model predictive control using a combination of explicit MPC and online optimization},
  author={Zeilinger, Melanie Nicole and Jones, Colin Neil and Morari, Manfred},
  journal={IEEE Trans. on Automatic Control},
  volume={56},
  number={7},
  pages={1524--1534},
  year={2011},
  publisher={IEEE}
}

@article{liao2020time,
  title={Time-distributed optimization for real-time model predictive control: Stability, robustness, and constraint satisfaction},
  author={Liao-McPherson, Dominic and Nicotra, Marco M and Kolmanovsky, Ilya},
  journal={Automatica},
  volume={117},
  pages={108973},
  year={2020},
  publisher={Elsevier}
}

@article{liao2021analysis,
  title={An analysis of closed-loop stability for linear model predictive control based on time-distributed optimization},
  author={Liao-McPherson, Dominic and Skibik, Terrence and Leung, Jordan and Kolmanovsky, Ilya and Nicotra, Marco M},
  journal={IEEE Trans. on Automatic Control},
  volume={67},
  number={5},
  pages={2618--2625},
  year={2021},
  publisher={IEEE}
}

@article{gros2020linear,
  title={From linear to nonlinear {MPC}: {Bridging} the gap via the real-time iteration},
  author={Gros, S{\'e}bastien and Zanon, Mario and Quirynen, Rien and Bemporad, Alberto and Diehl, Moritz},
  journal={International Journal of Control},
  volume={93},
  number={1},
  pages={62--80},
  year={2020},
  publisher={Taylor \& Francis}
}

@article{van2024suboptimal,
  title={Suboptimal {MPC} with a Computation Governor: Stability, Recursive Feasibility, and Applications to {ADMM}},
  author={van Leeuwen, Steven and Kolmanovsky, Ilya},
  journal={arXiv preprint arXiv:2411.07919},
  year={2024}
}

@article{yoshida2019instant,
  title={Instant {MPC} for linear systems and dissipativity-based stability analysis},
  author={Yoshida, Keisuke and Inoue, Masaki and Hatanaka, Takeshi},
  journal={IEEE Control Systems Letters},
  volume={3},
  number={4},
  pages={811--816},
  year={2019},
  publisher={IEEE}
}

@article{figura2024instant,
  title={Instant distributed {MPC} with reference governor},
  author={Figura, Martin and Su, Lanlan and Inoue, Masaki and Gupta, Vijay},
  journal={International Journal of Control},
  volume={97},
  number={4},
  pages={662--672},
  year={2024},
  publisher={Taylor \& Francis}
}

@article{moriyasu2024sampled,
  title={Sampled-Data Primal-Dual Gradient Dynamics in Model Predictive Control},
  author={Moriyasu, Ryuta and Kawaguchi, Sho and Kashima, Kenji},
  journal={Automatica},
  volume={183},
  pages={112621},
  year={2026},
  publisher={Elsevier}

}

@inproceedings{giselsson2010distributed,
  title={Distributed model predictive control with suboptimality and stability guarantees},
  author={Giselsson, Pontus and Rantzer, Anders},
  booktitle={IEEE Conference on Decision and Control},
  pages={7272--7277},
  year={2010}
}

@article{karapetyan2025closed,
  title={Closed-loop finite-time analysis of suboptimal online control},
  author={Karapetyan, Aren and Balta, Efe C and Iannelli, Andrea and Lygeros, John},
  journal={IEEE Trans. on Automatic Control},
  year={2025},
  publisher={IEEE}
}

@inproceedings{ames2019control,
  title={Control barrier functions: Theory and applications},
  author={Ames, Aaron D and Coogan, Samuel and Egerstedt, Magnus and Notomista, Gennaro and Sreenath, Koushil and Tabuada, Paulo},
  booktitle={European Control Conference},
  pages={3420--3431},
  year={2019},
}

@book{xiao2023safe,
  title={Safe Autonomy with Control Barrier Functions: Theory and Applications},
  author={Xiao, Wei and Cassandras, Christos G and Belta, Calin},
  year={2023},
  publisher={Springer}
}

@article{colombino2019online,
  title={Online optimization as a feedback controller: Stability and tracking},
  author={Colombino, Marcello and Dall’Anese, Emiliano and Bernstein, Andrey},
  journal={IEEE Trans. on Control of Network Systems},
  volume={7},
  number={1},
  pages={422--432},
  year={2019}
}

@article{hauswirth2020timescale,
  title={Timescale separation in autonomous optimization},
  author={Hauswirth, Adrian and Bolognani, Saverio and Hug, Gabriela and D{\"o}rfler, Florian},
  journal={IEEE Trans. on Automatic Control},
  volume={66},
  number={2},
  pages={611--624},
  year={2020}
}

@article{lohmiller1998contraction,
  title={On contraction analysis for non-linear systems},
  author={Lohmiller, Winfried and Slotine, Jean-Jacques E},
  journal={Automatica},
  volume={34},
  number={6},
  pages={683--696},
  year={1998},
  publisher={Elsevier}
}

@inproceedings{naghshtabrizi2006robust,
  title={On the robust stability and stabilization of sampled-data systems: A hybrid system approach},
  author={Naghshtabrizi, Payam and Hespanha, Joao P and Teel, Andrew R},
  booktitle={IEEE Conference on Decision and Control},
  pages={4873--4878},
  year={2006}
}

@article{fujioka2009discrete,
  title={A discrete-time approach to stability analysis of systems with aperiodic sample-and-hold devices},
  author={Fujioka, Hisaya},
  journal={IEEE Trans. on Automatic Control},
  volume={54},
  number={10},
  pages={2440--2445},
  year={2009}
}

@article{fagundes2023stability,
  title={Stability analysis of sampled-data systems with sector-bounded input nonlinearity},
  author={Fagundes, Arthur S and da Silva Jr, Jo{\~a}o M Gomes and Tarbouriech, Sophie},
  journal={IFAC-PapersOnLine},
  volume={56},
  number={2},
  pages={9812--9817},
  year={2023},
  publisher={Elsevier}
}

@article{sivashankar1993robust,
  title={Robust stability and performance analysis of sampled-data systems},
  author={Sivashankar, N and Khargonekar, Pramod P},
  journal={IEEE Trans. on Automatic Control},
  volume={38},
  number={1},
  pages={58--69},
  year={1993},
  publisher={IEEE}
}

@article{nesic2004framework,
  title={A framework for stabilization of nonlinear sampled-data systems based on their approximate discrete-time models},
  author={Nesic, Dragan and Teel, Andrew R},
  journal={IEEE Trans. on automatic control},
  volume={49},
  number={7},
  pages={1103--1122},
  year={2004},
  publisher={IEEE}
}

@article{gabriel2021sampled,
  title={Sampled-data control of {L}ur’e systems},
  author={Gabriel, Gabriela W and Geromel, Jos{\'e} C},
  journal={Nonlinear Analysis: Hybrid Systems},
  volume={40},
  pages={100994},
  year={2021},
  publisher={Elsevier}
}

\end{document}